\documentclass{article}

\usepackage[main, final]{neurips_2026}

\usepackage[utf8]{inputenc}%
\usepackage[T1]{fontenc}%
\usepackage{hyperref}%
\usepackage{url}%
\usepackage{booktabs}%
\usepackage{amsfonts}%
\usepackage{nicefrac}%
\usepackage{microtype}%
\usepackage{xcolor}%
\usepackage{mathtools}
\usepackage{cleveref}
\usepackage{multirow}
\usepackage{multicol}
\usepackage{makecell}
\usepackage{wrapfig}
\usepackage{subcaption}
\usepackage{float}
\usepackage{xcolor}
\usepackage[ruled,vlined]{algorithm2e}

\usepackage{amsthm}

\newtheorem{proposition}{Proposition}

\newif\ifhighlightchanges
\highlightchangestrue
\highlightchangesfalse%
\definecolor{newcolor}{RGB}{0,70,220}
\ifhighlightchanges
  \newcommand{\new}[1]{{\color{newcolor}#1}}
  \newenvironment{newtext}{\color{newcolor}}{}
\else
  \newcommand{\new}[1]{#1}
  \newenvironment{newtext}{}{}
\fi
\newcommand{\bluefloat}{\ifhighlightchanges\color{newcolor}\fi}
\newif\ifhighlightrevision
\highlightrevisiontrue
\highlightrevisionfalse%
\ifhighlightrevision
  \newcommand{\rev}[1]{{\color{newcolor}#1}}
\else
  \newcommand{\rev}[1]{#1}
\fi

\title{Harmonic Torsional Diffusion for Flexible Protein-Ligand Docking}

\author{%
  \new{Maksim Zhdanov$^{1,2}$\thanks{Correspondence to \texttt{dwdcodes@gmail.com}.}
  \quad Pavel Strashnov$^{2}$
  \quad Vladislav Kurenkov$^{2}$} \\[0.4em]
  \new{$^{1}$MBZUAI \qquad $^{2}$AXXX} \\
}

\begin{document}

\maketitle
\setcounter{footnote}{0}

\begin{abstract}  
Molecular docking requires reasoning jointly about ligand pose and protein flexibility. Most diffusion-based docking models predict torsional updates with generic Euclidean heads that ignore the periodic geometry of angular variables. This mismatch is especially limiting in flexible docking, where ligand conformations and pocket side chains co-adapt to form the bound complex. Here, we introduce Harmony, a harmonic torsional diffusion framework for flexible protein-ligand docking. Harmony parameterizes ligand and side-chain torsional score fields as derivatives of learned harmonic potentials on the circle, \rev{which are noised mode by mode in the Fourier basis owing to the properties of variance-exploding diffusion on the torus}. This construction makes periodicity explicit \new{and} gives the model a frequency-aware inductive bias over rotameric motion. On the PDBBind benchmark, Harmony improves ligand pose accuracy and pocket all-atom reconstruction over recent flexible docking methods. On PoseBusters, it improves the physical validity of generated complexes. Case studies on EBNA1 and KRAS G12D illustrate the method's behavior on a polar and a shallow binding site, respectively. Together, these results indicate that aligning the score parameterization with the geometry of the diffusion process is a simple and effective lever for improving flexible docking.
\end{abstract}

\section{Introduction}
\vspace{-0.5em}

Understanding binding phenomena between small molecules and proteins is central to modern drug discovery and computational methods aim to characterize these interactions through molecular docking \citep{pagadala2017software}. The field has evolved substantially, beginning with classical search-based strategies and more recently incorporating deep learning approaches. Despite these advances, a large portion of existing methods operates under the simplifying assumption that protein structures remain fixed during binding.\looseness=-1

This assumption contrasts with the inherently dynamic nature of proteins, which frequently adopt different conformations upon ligand association \citep{weikl2014conformational}. To address this, computational strategies that account for structural variability are typically divided into co-folding and flexible docking approaches \citep{abramson2024accurate, corso2025composing}. Co-folding attempts to determine the bound configuration of both protein and ligand simultaneously, treating the complex as a single prediction problem. In contrast, flexible docking focuses on capturing the conformational adjustments that occur between pre-existing unbound and bound states, offering advantages in efficiency, interpretability and controllability.\looseness=-1

Even so, achieving reliable accuracy with flexible docking remains a challenge. Traditional search-based methods are hindered by the expanded dimensionality introduced by protein flexibility, making exhaustive exploration impractical \citep{mcnutt2021gnina, koes2013lessons}. Deep learning-based methods attempt to mitigate this by integrating diffusion or flow-based frameworks with or without additional flexibility \citep{corso2022diffdock, plainer2023diffdock, corso2025composing}. However, they rely on the basic Riemannian generative machinery \new{\citep{jing2022torsional, de2022riemannian, chen2023flow}}, which does not utilize the rich geometric structure inherent in conformational degrees-of-freedom. In particular, protein side chains demonstrate specific conformations that are called rotamers, arising from rotations around $\chi$ dihedral (torsional) angles \new{\citep{branden2012introduction}}. The exact configuration of side chains is largely explained by steric repulsion between atoms whose spatial arrangement forms torsional angles. Because of this, most amino acids prefer to sit in discrete states in which they are positioned by some optimal dihedral angle value relative to the next groups. Importantly, ligand binding can further shift side chain rotamers by means of induced fit, reflecting a coupled relationship in which both ligand conformation and protein side-chain orientations adapt to optimize interactions \new{\citep{gaudreault2012side}}. Deep learning models for flexible docking should be able to naturally learn those multi-frequency motion features \citep{bahar2010global} to enable higher structural validity of predicted poses.\looseness=-1 

To address those challenges we propose \textit{Harmonic Torsional Diffusion} (in experimental section coined \textit{Harmony}), where we parametrize torsional angles of ligands and protein side chains in a Fourier basis, which comes with a clean analytical score function formulation and efficient noising structure which gradually suppresses frequency motion components coming from high to low harmonic variations.

We study this parametrization from the formal and empirical perspectives. We demonstrate that Harmony improves over recent generative methods for flexible docking on PDBBind benchmark \new{\citep{liu2017forging}} by \new{24.5 percentage points in ligand success rate over FlexDock} while following \new{a similar} training setup and model architecture. Moreover, \new{using the PoseBusters checks \citep{buttenschoen2024posebusters}}, Harmony shows improvements in chemical validity of generated poses over \new{DiffDock-Pocket and FlexDock}\rev{, and it also outperforms FlexDock on the PoseBusters benchmark set after OpenMM relaxation}. To demonstrate specific mechanism of our method empirically we conduct two case studies: EBNA1 and KRAS G12D. Additionally, we provide ablation studies on design choices behind \textit{Harmony}.

\begin{figure}[t]
\vspace{-\intextsep} 
  \centering
  \makebox[\linewidth][c]{%
    \includegraphics[width=1.00\linewidth]{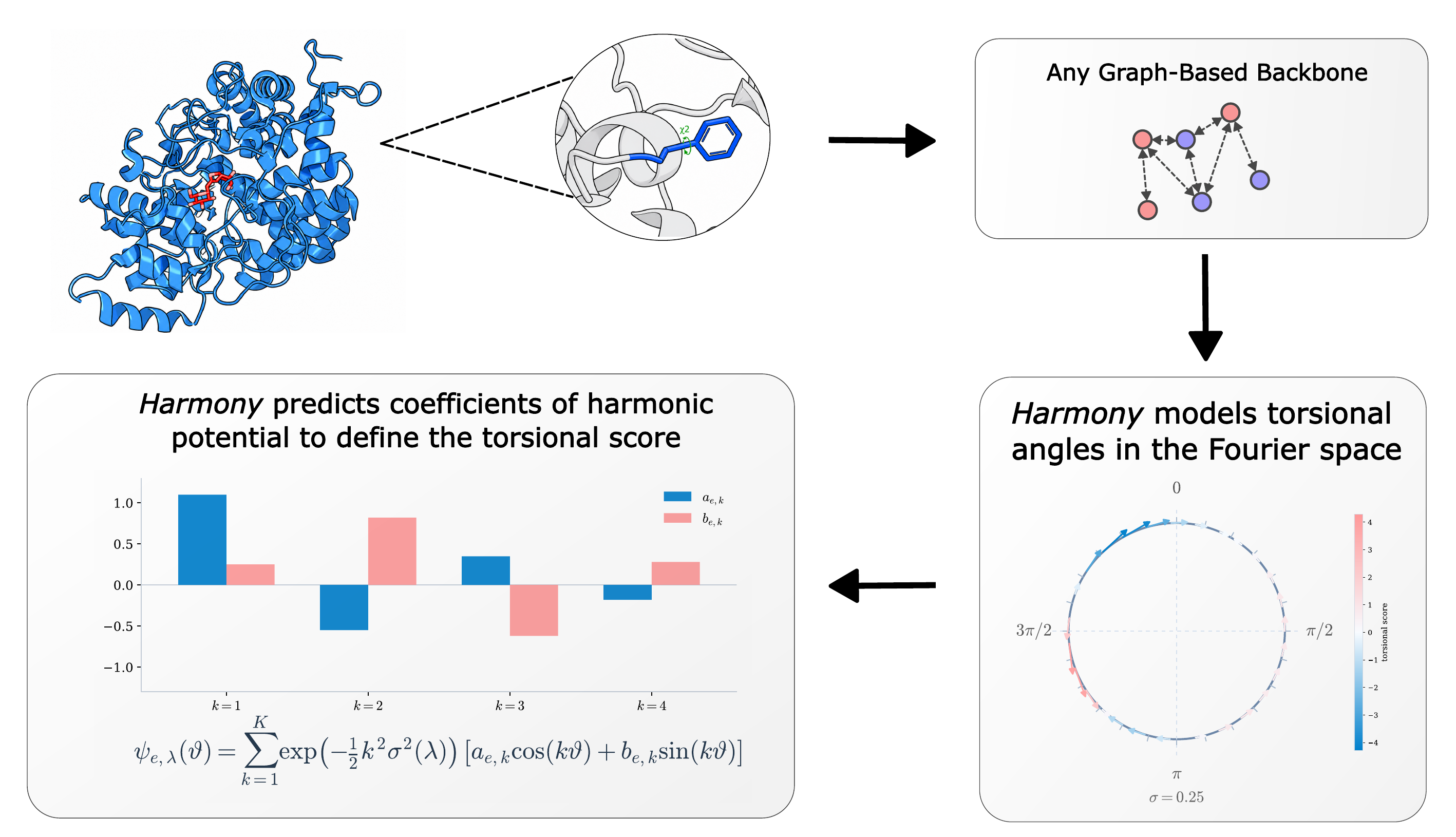}
  }
  \caption{\small
    Overview of \textit{Harmony}. Starting from a representation produced by any graph-based neural backbone, \textit{Harmony} replaces direct torsion regression with a structured periodic parameterization. Instead of predicting a single angular update, the model learns harmonic coefficients that define a Fourier potential, whose derivative yields the torsional score field for ligand and flexible side-chain dihedrals. This construction makes periodicity explicit and enables the model to naturally represent difficult torsional configurations. See more in \Cref{subsec:harmony}.
    }
  \label{fig:harmony-main}
    \vspace{-0.5em}
  \vspace{-\intextsep} 
\end{figure}

In summary, we propose Harmonic Torsional Diffusion, a new generative framework for flexible docking and modeling molecular conformations. We \rev{show empirically} that it naturally represents the specific degrees-of-freedom central to molecular binding. \new{The source code is publicly available.\footnote{\new{\url{https://github.com/dunnolab/harmony}}}}

\vspace{-0.5em}
\section{Related Work}
\vspace{-0.5em}

\paragraph{Flexible Molecular Docking.}

Flexible docking is typically formulated under the assumption that the unbound conformation of a protein (apo state) is available, with the objective of predicting the structural adjustments that lead to the bound complex (holo state). Classical docking approaches tackle this problem by defining an energy or scoring function and exploring the space of possible ligand poses to identify configurations that minimize this objective \citep{friesner2004glide, thomsen2006moldock, trott2010autodock}. In principle, protein flexibility, crucial to real-world applications \citep{teague2003implications}, can be incorporated by extending the search space to include additional degrees of freedom, such as side-chain torsional angles within the binding pocket. In practice, however, this significantly enlarges the search space, making it difficult to efficiently identify optimal joint configurations and often resulting in physically inconsistent or suboptimal poses.

More recently, deep learning-based approaches have been introduced to address these limitations. Methods such as DiffDock-Pocket \citep{plainer2023diffdock}, Re-Dock \citep{huang2024re}, FlexDock \citep{corso2025composing} extend generative frameworks, ranging from standard diffusion models to diffusion bridges and flow-based formulations, to account for flexibility in both ligands and protein binding pockets. Within this line of work, some methods restrict flexibility to side chains while treating the protein backbone as fixed, whereas others attempt to capture full structural variability, including backbone motion. Both perspectives can be justified depending on the biological context. \new{Related approaches include DiffBindFR \citep{zhu2024diffbindfr}, an all-atom $SE(3)$-equivariant diffusion model for flexible docking, FlowDock \citep{morehead2025flowdock}, which maps unbound to bound structures with geometric flow matching, and FIGRDock \citep{feng2025figrdock}, which replaces iterative generation with fast interaction-guided regression.} Recently, SigmaDock \citep{prat2025sigmadock} also explored the idea that the torsional module for docking is not optimal by reframing generative process with molecular fragmentation. However, they only focus on rigid pocket-based docking which is an easier subtask. 

The approach presented here departs from these existing methods by introducing a new generative framework specifically designed to naturally model torsional motion in both ligands and proteins. This \textit{harmonic torsional diffusion} formulation offers advantages in both efficiency and predictive performance. It is applicable in both pocket-based and blind docking settings and can accommodate either rigid or flexible protein representations. In the present work, the focus is restricted to side-chain flexibility, with a more detailed description of the task formulation provided in \Cref{sec:prelim}.

\vspace{-0.5em}
\paragraph{Diffusion Models.}
Score-based diffusion models \citep{song2020score} define a continuous diffusion process $dx=f(x,t)dt+g(t)dw$ that transports the data distribution $p_0$ \new{into} a simple prior $p_T$. This has the corresponding reverse SDE $dx=[f(x,t)-g(t)^2\nabla_x \log p_t(x)]dt+g(t)dw$ where only the score $\nabla_x \log p_t(x)$ is unknown. Using denoising score matching \citep{vincent2011connection}, one can learn $s(x,t) \approx \nabla_x \log p_t(x)$ and use it to run the reverse SDE to obtain samples from the data distribution. We base \textit{Harmony's} noising process on Variance-Exploding (VE) SDE and further demonstrate its relation to the circular nature of torsional angles.

\vspace{-0.5em}
\section{Preliminaries} \label{sec:prelim}

\vspace{-0.5em}
\paragraph{All-Atom Ligand--Protein Representation.}
We represent each complex as a heterogeneous graph
\[
\mathcal{G}=\{\mathcal{V}\coloneqq(\mathcal{V}^l,\mathcal{V}^b,\mathcal{V}^a),\;
\mathcal{E}\coloneqq(\mathcal{E}^l,\mathcal{E}^b,\mathcal{E}^a,\mathcal{E}^{lb},\mathcal{E}^{la},\mathcal{E}^{ab})\},
\]
where \(\mathcal{V}^l\) denotes ligand atoms, \(\mathcal{V}^b\) residue-level receptor backbone nodes at \(C_\alpha\) positions, and \(\mathcal{V}^a\) receptor atoms including both backbone and side-chain atoms. Ligand nodes carry standard atom-level chemical descriptors, residue nodes encode amino-acid identity together with a pretrained ESM-2 embedding \citep{lin2022language}, and receptor atom nodes encode residue identity and atom-type information. The edge sets \(\mathcal{E}^l,\mathcal{E}^b,\mathcal{E}^a\) capture intra-graph ligand, residue, and atom interactions, while \(\mathcal{E}^{lb},\mathcal{E}^{la},\mathcal{E}^{ab}\) couple ligand atoms to receptor residues, ligand atoms to receptor atoms, and receptor atoms to their parent residues. This multi-scale construction allows the model to reason jointly over ligand chemistry, residue-level geometry, and full atomic structure. Detailed feature definitions and graph construction are given in Appendix \ref{app:graph}.

\vspace{-0.5em}
\paragraph{Pocket-Based Flexible Docking.}
To focus modeling on the binding interface, we replace the full receptor with a pocket-centered subgraph \(\mathcal{G}_{\mathrm{pocket}}\) containing the ligand together with nearby receptor residues and atoms, following standard pocket-reduction practice \citep{corso2025composing}. We keep the protein backbone fixed and allow motion only in pocket side chains, since side-chain rearrangements near the binding site account for much of the local conformational adaptation required for docking \citep{miao2016quantifying,clark2019inherent,alberts2005receptor}. The docking state is therefore parameterized as
\[
\mathbf{z}=(\mathbf{t},\mathbf{R},\boldsymbol{\tau},\boldsymbol{\chi})
\in
\mathcal{M}
=
\mathbb{R}^3 \times SO(3) \times \mathbb{T}^{m_l} \times \mathbb{T}^{m_s},
\]
where \(\mathbf{t}\in\mathbb{R}^3\) and \(\mathbf{R}\in SO(3)\) denote ligand translation and rotation, \(\boldsymbol{\tau}\in\mathbb{T}^{m_l}\) denotes ligand torsions, and \(\boldsymbol{\chi}\in\mathbb{T}^{m_s}\) denotes flexible pocket side-chain torsions. Thus, docking is modeled on a product manifold with Euclidean, rotational, and toroidal components, and the network jointly predicts ligand pose and pocket side-chain arrangement from \(\mathcal{G}_{\mathrm{pocket}}\). Exact pocket selection details are deferred to Appendix~\new{\ref{app:data-prep}}.

\vspace{-0.5em}
\section{Harmonic Torsional Diffusion}
\vspace{-0.5em}

In this section, we firstly build the foundation of the \textit{Harmonic Torsional Diffusion} (\textit{Harmony}) and then derive explicit parametrization and training for the model. All the proofs for the relevant propositions can be found in Appendix \ref{app:proofs}.

\subsection{Variance-Exploding SDE} \label{subsec:ve-sde}

We adopt the variance-exploding (VE) diffusion SDE framework as the probabilistic foundation of our method. In VE diffusion, the forward process perturbs data through pure stochastic diffusion, without any deterministic drift term. In Euclidean space, the forward SDE takes the form
\[
\mathrm{d}\mathbf{z}_\lambda = g(\lambda)\,\mathrm{d}\mathbf{w}_\lambda,
\]
where $\mathbf{w}_\lambda$ is standard Brownian motion and $g(\lambda)$ is a time-dependent diffusion coefficient. Unlike variance-preserving (VP) SDE, VE-SDE diffusion models do not contract the state toward the origin during the forward process. Instead, the variance increases monotonically with time, gradually transforming the data distribution into a prior distribution.

The corresponding marginal perturbation kernel is Gaussian, $p_\lambda(\mathbf{z}_\lambda \mid \mathbf{z}_0)
=
\mathcal{N}\!\left(
\mathbf{z}_0,\,
\sigma(\lambda)^2 \mathbf{I}
\right)$,
where $\sigma(\lambda)$ is the noise scale induced by the diffusion coefficient $g(\lambda)$. In practice, we use an exponential noise schedule, $\sigma(\lambda)=\sigma_{\min}^{1-\lambda}\sigma_{\max}^{\lambda}$.

A key property of the VE forward process is that its density evolution satisfies a heat equation as demonstrated in \Cref{prop:ve-heat}. The proposition \new{below} suggests a natural parameterization for torsional variables. Since each torsion angle lives on the circle $\mathbb{T}$ and the VE forward process acts as heat flow, the appropriate basis is given by the eigenfunctions of the Laplace-Beltrami operator on $\mathbb{T}$ which are the Fourier modes. We elaborate on the main idea in the following subsections.

\begin{proposition}[VE-SDE as heat flow, cf.\ \citet{sarkar2026score}] \label{prop:ve-heat}
Let the forward process be defined by the driftless stochastic differential equation
\[
\mathrm{d}\mathbf{z}_\lambda = g(\lambda)\,\mathrm{d}\mathbf{w}_\lambda,
\qquad
\mathbf{z}_0 \sim p_0,
\]
where $\mathbf{w}_\lambda$ is standard Brownian motion and $g(\lambda)$ is a time-dependent diffusion coefficient. Then the density $p_\lambda(\mathbf{z})$ of $\mathbf{z}_\lambda$ satisfies the Fokker--Planck equation
\[
\frac{\partial p_\lambda(\mathbf{z})}{\partial \lambda}
=
\frac{g(\lambda)^2}{2}\,\Delta p_\lambda(\mathbf{z}),
\]
that is, a heat equation with time-dependent diffusivity
\[
\nu(\lambda)=\frac{g(\lambda)^2}{2}.
\]
\end{proposition}

\subsection{Product Space Diffusion} \label{subsec:product-diff}
We model flexible docking with a variance-exploding (VE) diffusion process on the product space of ligand rigid-body motion and periodic torsional variables. Since the protein backbone is kept fixed and only pocket side chains are allowed to move, the forward process as defined according to $\mathbf{z}$, perturbs the translation, rotation, ligand torsions and pocket side-chain torsions independently at noise level $\lambda \in [0,1]$:
\[
q_\lambda(\mathbf{z}_\lambda \mid \mathbf{z}_0)
=
q_\lambda^{\mathrm{tr}}(\mathbf{t}_\lambda \mid \mathbf{t}_0)\,
q_\lambda^{\mathrm{rot}}(\mathbf{R}_\lambda \mid \mathbf{R}_0)\,
q_\lambda^{\mathrm{tor}}(\boldsymbol{\tau}_\lambda \mid \boldsymbol{\tau}_0)\,
q_\lambda^{\mathrm{sc}}(\boldsymbol{\chi}_\lambda \mid \boldsymbol{\chi}_0),
\]
where each component is corrupted with an exponentially increasing noise scale
\[
\sigma_u(\lambda)=\sigma_{u,\min}^{1-\lambda}\sigma_{u,\max}^{\lambda},
\qquad
u \in \{\mathrm{tr},\mathrm{rot},\mathrm{tor},\mathrm{sc}\}.
\]

For ligand translation, we apply Gaussian perturbations in Euclidean space,
\[
\mathbf{t}_\lambda = \mathbf{t}_0 + \sigma_{\mathrm{tr}}(\lambda)\,\boldsymbol{\epsilon}_{\mathrm{tr}},
\qquad
\boldsymbol{\epsilon}_{\mathrm{tr}} \sim \mathcal{N}(\mathbf{0}, \mathbf{I}_3).
\]
For ligand rotation, we perturb the clean orientation on $SO(3)$ with isotropic rotational noise. In practice, this perturbation is represented in axis-angle form and applied through the exponential map,
\[
\mathbf{R}_\lambda = \mathrm{Exp}(\boldsymbol{\omega}_\lambda)\,\mathbf{R}_0,
\qquad
\boldsymbol{\omega}_\lambda \sim \mathrm{IGSO}(3;\sigma_{\mathrm{rot}}(\lambda)).
\]
For ligand torsions and pocket side-chain torsions, we use wrapped Gaussian perturbations on the torus,
\[
\boldsymbol{\tau}_\lambda
=
\mathrm{wrap}\!\left(
\boldsymbol{\tau}_0 + \sigma_{\mathrm{tor}}(\lambda)\,\boldsymbol{\epsilon}_{\mathrm{tor}}
\right),
\qquad
\boldsymbol{\epsilon}_{\mathrm{tor}} \sim \mathcal{N}(\mathbf{0}, \mathbf{I}),
\]
\[
\boldsymbol{\chi}_\lambda
=
\mathrm{wrap}\!\left(
\boldsymbol{\chi}_0 + \sigma_{\mathrm{sc}}(\lambda)\,\boldsymbol{\epsilon}_{\mathrm{sc}}
\right),
\qquad
\boldsymbol{\epsilon}_{\mathrm{sc}} \sim \mathcal{N}(\mathbf{0}, \mathbf{I}),
\]
where $\mathrm{wrap}(\cdot)$ maps angles back to $(-\pi,\pi]$.

The network is trained to predict the score of each factor of the forward noising distribution. Reverse-time sampling then jointly denoises these components to recover the bound ligand pose together with the binding-pocket side-chain arrangement according to backward process \citep{anderson1982reverse}.

\subsection{Harmonic Parametrization of Torsional Angles} \label{subsec:harmony}

Overall design of \textit{Harmony} can be seen in \Cref{fig:harmony-main}. The VE-SDE viewpoint above motivates a natural representation for torsional variables. Since ligand torsions $\boldsymbol{\tau}$ and side-chain torsions $\boldsymbol{\chi}$ take values on a torus, each individual torsional degree of freedom lies on the circle $\mathbb{T}$. Two facts make this geometry tractable.

\begin{proposition}[Fourier eigenbasis of the circle, \citet{stein2003fourier}] \label{prop:fourier-eigenbasis}
Let $\Delta_{\mathbb{T}}$ denote the Laplace-Beltrami operator on the circle $\mathbb{T}$. Then, for every integer $k \ge 1$, the functions $\cos(k\vartheta)$ and $\sin(k\vartheta)$ are eigenfunctions of $\Delta_{\mathbb{T}}$ with eigenvalue $-k^2$, i.e.
\[
\Delta_{\mathbb{T}}\cos(k\vartheta) = -k^2\cos(k\vartheta),
\qquad
\Delta_{\mathbb{T}}\sin(k\vartheta) = -k^2\sin(k\vartheta).
\]
Consequently, for any finite Fourier expansion
\[
\psi(\vartheta)=\sum_{k=1}^K \big[a_k\cos(k\vartheta)+b_k\sin(k\vartheta)\big],
\]
the heat semigroup generated by $\Delta_{\mathbb{T}}$ acts diagonally:
\[
\exp\!\left(\frac{\sigma^2}{2}\Delta_{\mathbb{T}}\right)\psi(\vartheta)
=
\sum_{k=1}^K
\exp\!\left(-\frac{1}{2}k^2\sigma^2\right)
\big[a_k\cos(k\vartheta)+b_k\sin(k\vartheta)\big].
\]
\end{proposition}

\begin{proposition}[Wrapped Gaussian noising is heat flow on $\mathbb{T}$, {\citet[\S 3.5.7]{mardia2009directional}}]
\label{prop:wrapped-heat}
Let $\vartheta_\lambda = \mathrm{wrap}(\vartheta_0 + \sigma(\lambda)\,\epsilon)$ with $\epsilon \sim \mathcal{N}(0,1)$. Then the conditional density of $\vartheta_\lambda$ given $\vartheta_0$ is the wrapped normal
\begin{equation*}
p_\lambda(\vartheta \mid \vartheta_0) = \sum_{n \in \mathbb{Z}} \frac{1}{\sqrt{2\pi\sigma^2(\lambda)}} \exp\!\left(-\frac{(\vartheta - \vartheta_0 - 2\pi n)^2}{2\sigma^2(\lambda)}\right),
\end{equation*}
which coincides with the heat kernel on $\mathbb{T}$ at time $\sigma^2(\lambda)/2$. Equivalently, $p_\lambda(\,\cdot\mid\vartheta_0) = \exp\!\left(\tfrac{\sigma^2(\lambda)}{2}\Delta_\mathbb{T}\right)\delta_{\vartheta_0}$, so the marginal density satisfies $\partial_\lambda p_\lambda = \tfrac{g(\lambda)^2}{2}\Delta_\mathbb{T}\, p_\lambda$.
\end{proposition}

Let $\vartheta$ denote one component of either $\boldsymbol{\tau}$ or $\boldsymbol{\chi}$ associated with a rotatable bond $e$. For each such bond, the message-passing backbone produces a local bond representation $\mathbf{u}_e$, from which \textit{Harmony} predicts harmonic coefficients via
\[
[a_{e,1}, b_{e,1}, \dots, a_{e,K}, b_{e,K}]
=
\mathbf{g}_\theta(\mathbf{u}_e),
\]
where $K$ is the number of retained Fourier modes. These coefficients define a latent harmonic function on the circle,
\[
\psi_e(\vartheta)
=
\sum_{k=1}^K
\left[
a_{e,k}\cos(k\vartheta)
+
b_{e,k}\sin(k\vartheta)
\right].
\]

By Propositions~\ref{prop:ve-heat}-\ref{prop:wrapped-heat}, the VE forward process on each torsional coordinate acts as heat flow on $\mathbb{T}$, and Proposition~\ref{prop:fourier-eigenbasis} diagonalizes this action in the Fourier basis. The smoothed harmonic function at noise level $\lambda$ is therefore
\begin{equation}
\psi_{e,\lambda}(\vartheta)
=
\sum_{k=1}^K
\exp\!\left(-\frac{1}{2}k^2\sigma^2(\lambda)\right)
\left[
a_{e,k}\cos(k\vartheta)
+
b_{e,k}\sin(k\vartheta)
\right]
\label{eq:psi}
\end{equation}
Thus, VE diffusion induces a frequency-dependent damping, where high-frequency angular structure is suppressed at large noise, while finer modes gradually re-emerge as $\lambda \to 0$ as in \Cref{fig:damping}.

\textit{Harmony} predicts the torsional score field analytically by differentiating the heat-smoothed harmonic expansion with respect to the angle:
\[
s_e(\vartheta,\lambda)
=
\partial_\vartheta \psi_{e,\lambda}(\vartheta)
=
\sum_{k=1}^K
k\exp\!\left(-\frac{1}{2}k^2\new{\sigma^2(\lambda)}\right)
\left[
-a_{e,k}\sin(k\vartheta)
+
b_{e,k}\cos(k\vartheta)
\right].
\]
Applying this construction independently to all ligand and side-chain rotatable bonds yields the torsional score components associated with $\boldsymbol{\tau}$ and $\boldsymbol{\chi}$.

This parameterization is crucial because the model does not need to learn an arbitrary score field from scratch. Instead, \new{it applies} the heat-semigroup damping factor \new{$\exp(-\frac{1}{2}k^2\sigma^2(\lambda))$}, which is exactly how the forward process transforms Fourier modes on $\mathbb{T}$. As a result, the representation remains faithful to the noising process across all noise levels, preserves periodicity by construction and naturally captures multimodal torsional landscapes such as ligand conformers and side-chain rotamers.

\vspace{-0.5em}
\subsection{Confidence and Relaxation Modules}
\label{subsec:conf-aff-relax}
\begin{wrapfigure}[20]{r}{0.48\textwidth}
  \vspace{-1.2em}
  \centering
  \includegraphics[width=0.46\textwidth]{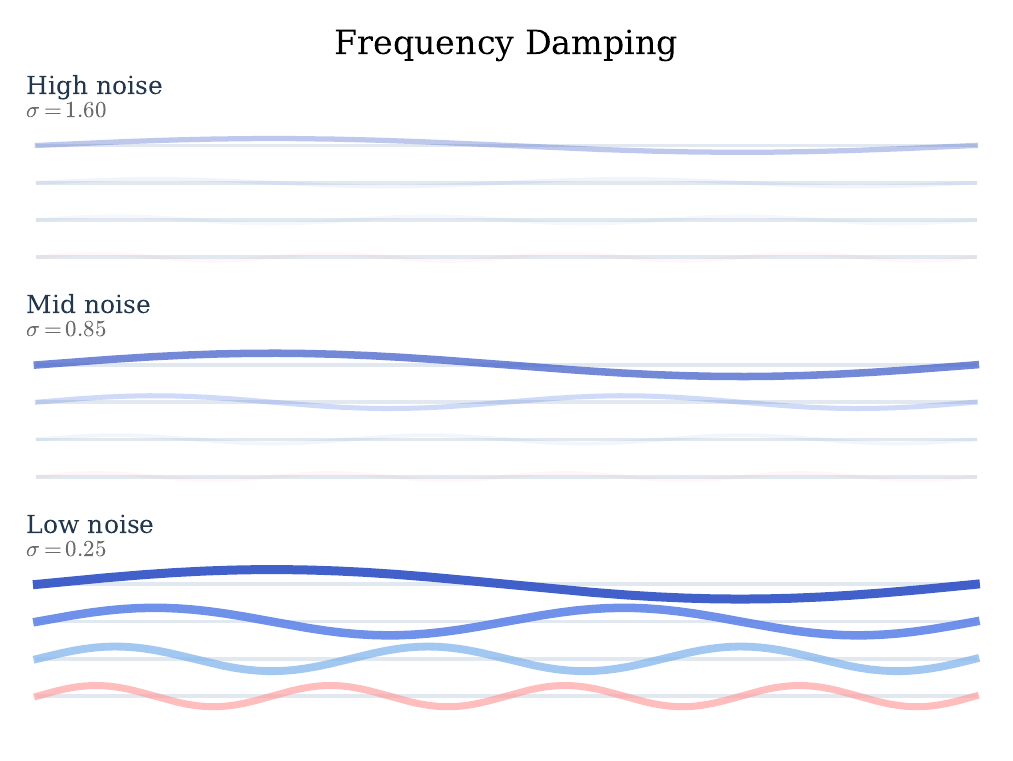}
  \vspace{-1.4em}
  \caption{\small
    Frequency damping effect of VE noising. As noise level increases high-frequency components saturate faster, but low-frequency modes remain. Thus, \textit{Harmony} learns torsional motion naturally by filtering out different scales of molecular motion during training.
  }
  \label{fig:damping}
\end{wrapfigure}

\vspace{-0.5em}
\paragraph{Confidence Head.}
We use a lightweight confidence head attached to the shared docking encoder rather than a separate reranking model. The head predicts a scalar confidence score for each sampled pose and is supervised with protein-ligand interface lDDT (PLI-lDDT) \citep{lu2024dynamicbind}, which measures local receptor--ligand contact accuracy. To emphasize near-final geometries, confidence is trained with a time-weighted regression loss that upweights low-noise samples. More details in \Cref{app:model}.

\paragraph{Relaxation.}
We do not apply post hoc relaxation \new{in our main results}. This choice allows us to isolate the quality of the manifold docking model itself and evaluate the generated complexes directly. In principle, \textit{Harmony} can be combined with any downstream relaxation procedure, including separately trained relaxation models similar to FlexDock \citep{corso2025composing} or external tools such as GNINA \citep{mcnutt2021gnina}. \new{Results after OpenMM \citep{eastman2023openmm} relaxation are reported in Appendix~\ref{app:relaxation}.}\looseness=-1

\paragraph{Model Architecture and Training.}
\textit{Harmony} uses an e3nn \citep{geiger2022e3nn} heterogeneous $SE(3)$-equivariant graph neural network over the ligand atom graph, the receptor $C_\alpha$ graph, and the receptor atom graph in the binding pocket. We train the model with the standard diffusion score-matching objective as in \citep{plainer2023diffdock}. Training uses OpenEquivariance \citep{bharadwaj2025efficient} kernels for Clebsch-Gordan tensor products, reducing runtime from roughly $4.5$ days to $2$ days on $8$ H100 80GB GPUs (see Appendix \ref{app:model}).

\vspace{-1.0em}
\section{Experiments}

\begin{table}[t]
\vspace{-\intextsep} 
\centering
\caption{\small \textbf{Top-1 PDBBind ESMFold Docking Performance}. \textit{Harmony} samples 10 complexes and selects one best on the highest confidence. \new{For \textit{Harmony}, DiffDock-Pocket, FIGRDock and DiffBindFR, we report mean\,$\pm$\,std over three training seeds. ReDock and FlexDock results are taken from prior work.} The best results are highlighted \textbf{in bold}. N.A. indicates results missing due to baselines taken from previous work where output structures were not provided. Runtime of \textit{Harmony} was measured for regular Clebsch-Gordan tensor product implementation and with OpenEquivariance CUDA kernels utilized in our base model.}
\vspace{0.5em}

\setlength{\tabcolsep}{6pt}
\renewcommand{\arraystretch}{1.1}

\begin{tabular}{lcccc}
\toprule
\multirow{2}{*}{Method} 
& \multicolumn{2}{c}{Ligand RMSD} 
& All-Atom 
& \multirow{2}{*}{Runtime (s)} \\
\cmidrule(lr){2-3} \cmidrule(lr){4-4}
& \shortstack{\%$<2$\AA{}$\uparrow$} 
& \shortstack{Med.\ \AA{}$\downarrow$} 
& \shortstack{\%$<1$\AA{}$\uparrow$} 
& \\
\midrule
SMINA (rigid)           & 6.6  & 7.7 & N.A. & 258  \\
SMINA                   & 3.6  & 7.3 & 5.2  & 1914 \\
GNINA (rigid)           & 6.7  & 7.1 & N.A. & 260  \\
GNINA                   & 8.4  & 7.9 & 4.5  & 1575 \\
\midrule
DiffDock-Pocket   & 41.8\new{\,$\pm$\,1.7} & 2.5 & 32.4\new{\,$\pm$\,1.5} & 17   \\
ReDock            & 39.0 & 2.5 & 39.8 & 15   \\
FlexDock          & 39.7 & 2.5 & 41.7 & 11   \\
\new{FIGRDock}    & \new{46.2\,$\pm$\,1.0} & \new{2.4} & \new{44.3\,$\pm$\,0.5} & \new{\textbf{0.5}} \\
\new{DiffBindFR}  & \new{51.0\,$\pm$\,0.3} & \rev{2.1} & \new{36.7\,$\pm$\,1.2} & \new{16} \\
\midrule
\textit{Harmony}  & \new{\textbf{64.2\,$\pm$\,0.7}} & \new{\textbf{1.9}} & \new{\textbf{47.3\,$\pm$\,1.1}} & \new{5 (11)} \\
\bottomrule
\end{tabular}

\label{tab:docking_results}
\vspace{-\intextsep} 
\end{table}

\vspace{-0.5em}
\subsection{Experimental Setup and Results} \label{subsec:exp-main}

\vspace{-0.5em}
\paragraph{Benchmarks.}
We evaluate \textit{Harmony} using the widely established PDBBind benchmark \new{\citep{liu2017forging}}, relying on structures predicted by ESMFold \new{\citep{lin2022language}} to approximate the distribution of unbound protein conformations. Given the focus on flexible docking, the evaluation considers the accuracy of both ligand placement and binding pocket geometry. All of the reported metrics are for top-1 samples ranked by confidence and do not use oracle information. More details on datasets are in Appendix \ref{app:data}.\looseness=-1

To ensure a consistent comparison, predicted pocket structures are first aligned to their corresponding reference conformations. Following this alignment, heavy-atom RMSD is computed for both the ligand and the pocket atoms. Performance for ligand predictions is summarized using the median RMSD as well as the proportion of cases below a 2\AA{} threshold. For pocket atoms, accuracy is reported as the fraction of predictions achieving an RMSD below 1\AA{}. More details on evaluation metrics can be found in Appendix \ref{app:eval-metrics}.
Additionally, we \new{assess the physical validity of generated poses with} the PoseBusters \citep{buttenschoen2024posebusters} \new{checks, which include} a wide range of filters \new{for chemical and geometric plausibility}. Exact stratification between different validity checks can be seen in Appendix \ref{app:strat-posebusters}\looseness=-1

\begin{wrapfigure}[14]{r}{0.48\linewidth}
\vspace{-1.6em}
  \centering
  \includegraphics[width=\linewidth]{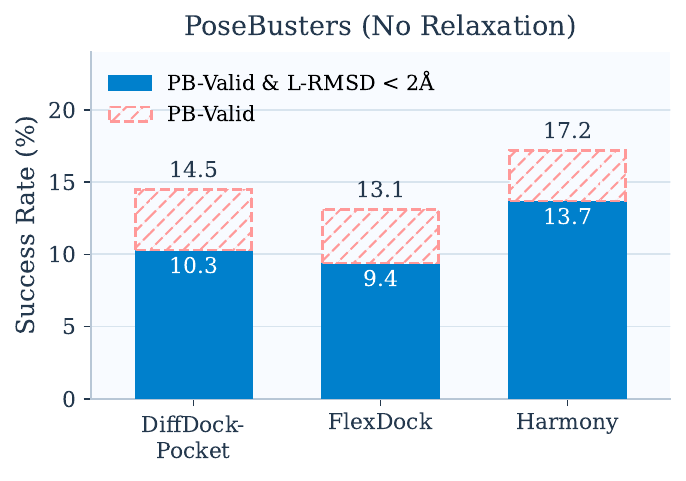}
  \vspace{-1.6em}
  \caption{\small \new{PoseBusters validity on the PDBBind test set} with \textbf{no relaxation}.}
  \label{fig:posebusters}
  \vspace{-1.0em}
\end{wrapfigure}

\vspace{-1.0em}
\paragraph{Baselines.}
On the PDBBind benchmark, \textit{Harmony} is evaluated against both classical and learning-based baselines. We include established search-based docking methods such as SMINA \citep{koes2013lessons} and GNINA \citep{mcnutt2021gnina}, as well as recent flexible, pocket-level machine learning approaches including DiffDock-Pocket \new{\citep{plainer2023diffdock}}, ReDock \new{\citep{huang2024re}}\new{,} FlexDock \new{\citep{corso2025composing}}\new{, FIGRDock \citep{feng2025figrdock} and DiffBindFR \citep{zhu2024diffbindfr}}.

A notable consideration in these comparisons is the use of post-processing or relaxation steps. Several methods incorporate external energy minimization procedures which are often implemented via tools such as OpenMM or rely on auxiliary learned models to refine predicted complexes as in FlexDock. While such procedures can improve geometric plausibility by reducing steric clashes and energy artifacts, they may also obscure deficiencies in the underlying generative model by correcting both inter- and intramolecular inconsistencies.

To isolate the intrinsic performance of flexible docking models, we conduct additional comparisons \new{with the PoseBusters checks on the PDBBind test set} using base versions of DiffDock-Pocket and recent FlexDock, explicitly excluding any relaxation or post-processing components. This allows for a more direct assessment of the generative models themselves, independent of downstream correction mechanisms.

\paragraph{Main Results.}
\Cref{tab:docking_results} summarizes the performance of prior approaches alongside \textit{Harmony}, where consistent improvements are observed across multiple evaluation metrics. \textit{Harmony} adopts \new{a similar} architectural design and training protocol \new{to} DiffDock-Pocket and FlexDock. Under this setting, we observe gains in both ligand placement accuracy, measured by the fraction of predictions with RMSD < 2\AA{} with the increase from 39.7\% to \new{64.2\%}, and pocket reconstruction quality where AA-RMSD < 1\AA{} increases from 41.7\% to \new{47.3\%. The margin over FlexDock persists for Top-5 and Oracle-10 selection, \rev{which depend less or not at all on} the confidence model (Appendix~\ref{app:extended-baselines}). Since RMSD-based metrics measure the angular variables modeled by \textit{Harmony} only indirectly, we also evaluate torsion-level accuracy, namely the circular error of ligand and side-chain dihedrals and side-chain rotamer recovery, where \textit{Harmony} improves over FlexDock on all metrics (Appendix~\ref{app:torsion-metrics}).}\looseness=-1

An important distinction is that FlexDock relies on an additional relaxation module, which must be trained separately and applied as a post-processing step. In contrast, the model proposed here directly generates final complex structures without any auxiliary refinement, demonstrating that the improvements stem from the core generative framework rather than downstream correction procedures.
On \new{PoseBusters validity} (\Cref{fig:posebusters}), our method outperforms direct pocket-based flexible competitors DiffDock-Pocket and FlexDock.\looseness=-1

Overall, these results suggest that the Harmonic Torsional Diffusion framework introduces effective inductive biases for jointly modeling ligand placement and side-chain conformations, leading to improved structural accuracy.

\begin{figure}[t]
\vspace{-\intextsep} 
  \centering
  \makebox[\linewidth][c]{%
    \includegraphics[width=1.00\linewidth]{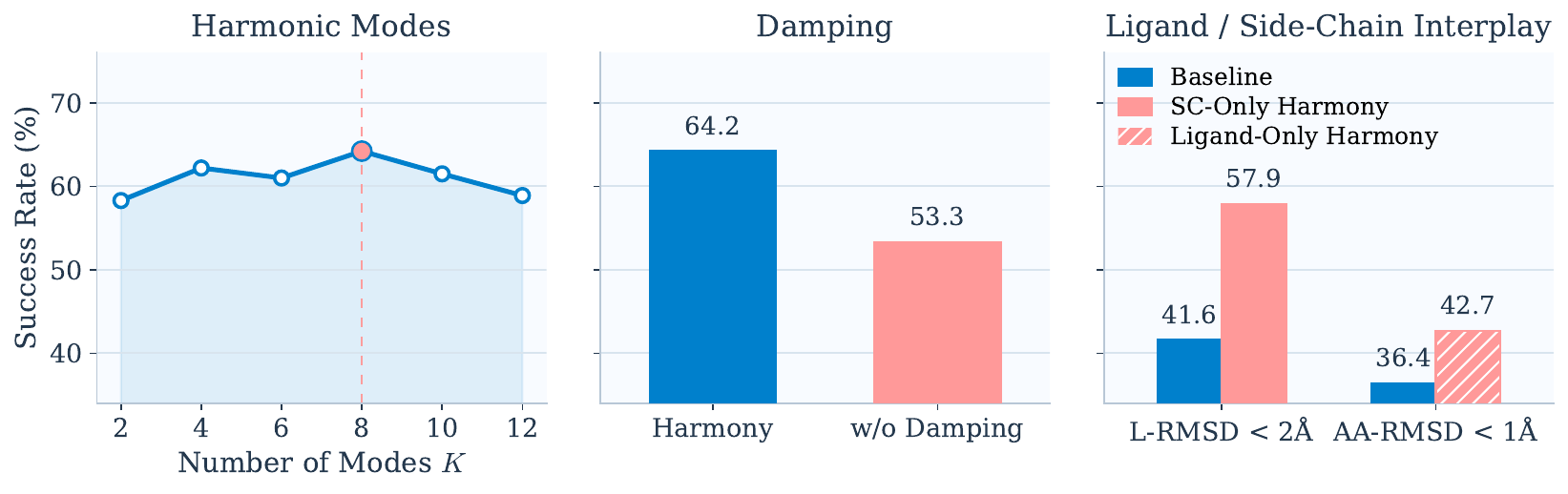}
  }
  \caption{\small
    Ablations of \textit{Harmony}. \textbf{Left}: Ablation of the number of harmonic modes $K$ on PDBBind success rate. \new{A moderate number of modes gives the best performance.} \textbf{Middle}: Ablation of the VE heat-flow damping term. Removing damping degrades docking success, showing that explicit noise-dependent suppression of high-frequency torsional modes is important. \textbf{Right}: Applying harmonic learning only to ligand torsions improves side-chain modeling, while applying it only to side chains improves ligand success rate. This indicates that Harmony improves the joint representation used for flexible docking.\looseness=-1
    }

  \label{fig:harmony-three-ablations}
  \vspace{-\intextsep} 
\end{figure}

\begin{figure}[t]
\vspace{-\intextsep} 
  \centering

  \begin{tabular}{cc}
    \textbf{Baseline} & \textbf{\textit{Harmony}} \\[0.4em]

    \includegraphics[width=0.45\linewidth]{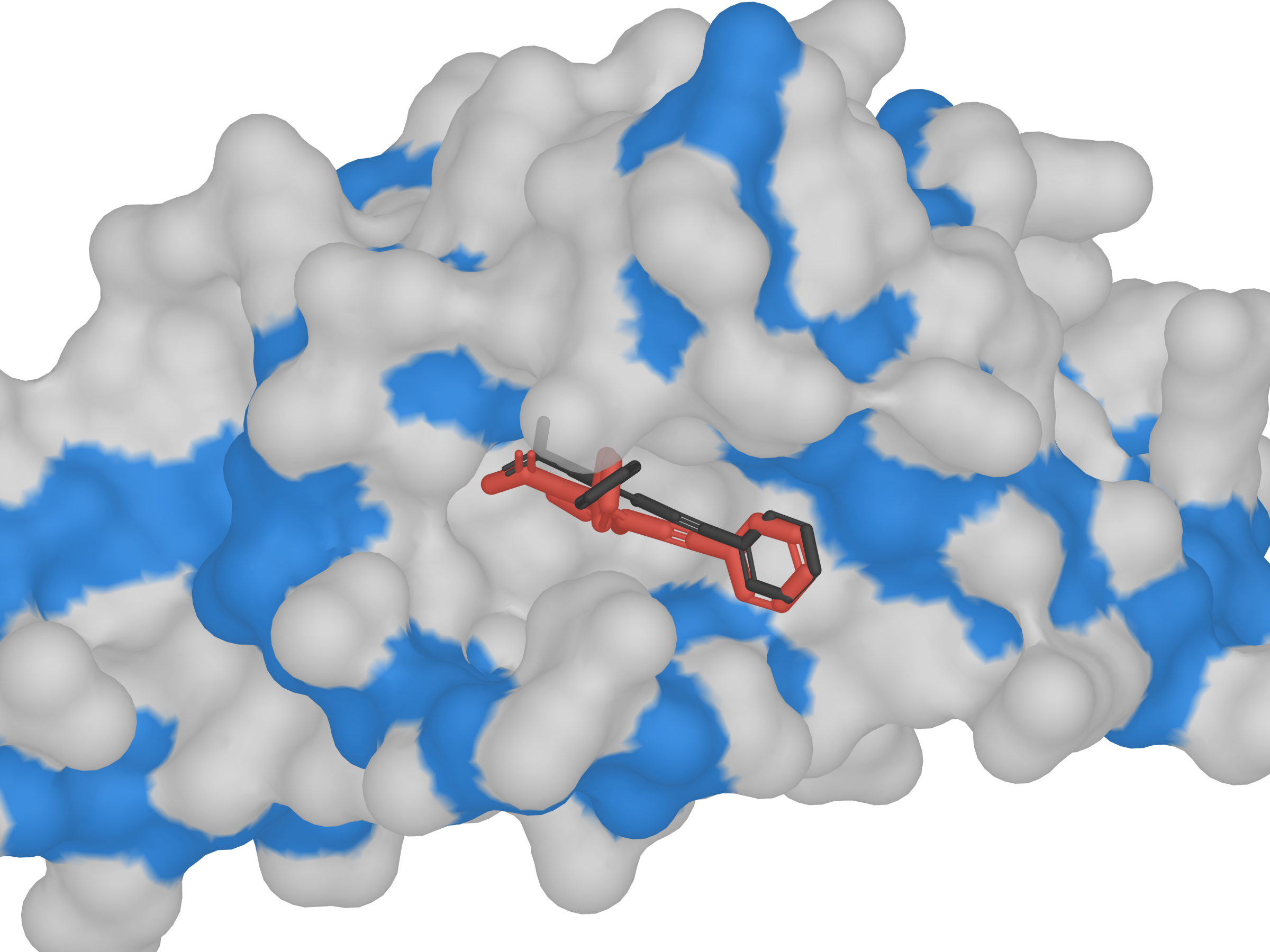} &
    \includegraphics[width=0.45\linewidth]{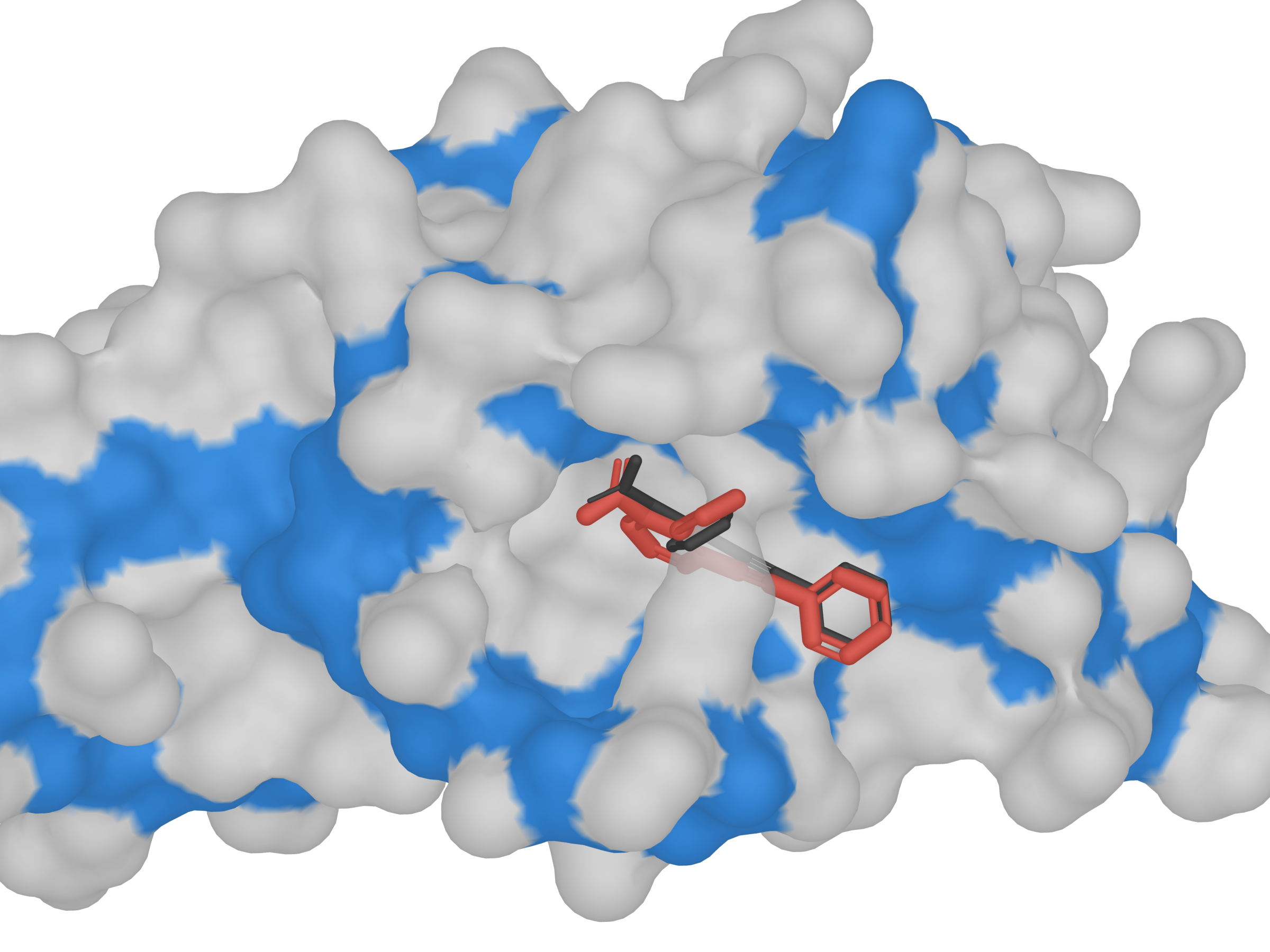} \\
    \small \textbf{6NPP}: L-RMSD \textbf{1.35 \AA}, AA-RMSD \textbf{1.79 \AA} &
    \small \textbf{6NPP}: L-RMSD \textbf{0.72 \AA}, AA-RMSD \textbf{1.46 \AA} \\[1.0em]

    \includegraphics[width=0.45\linewidth]{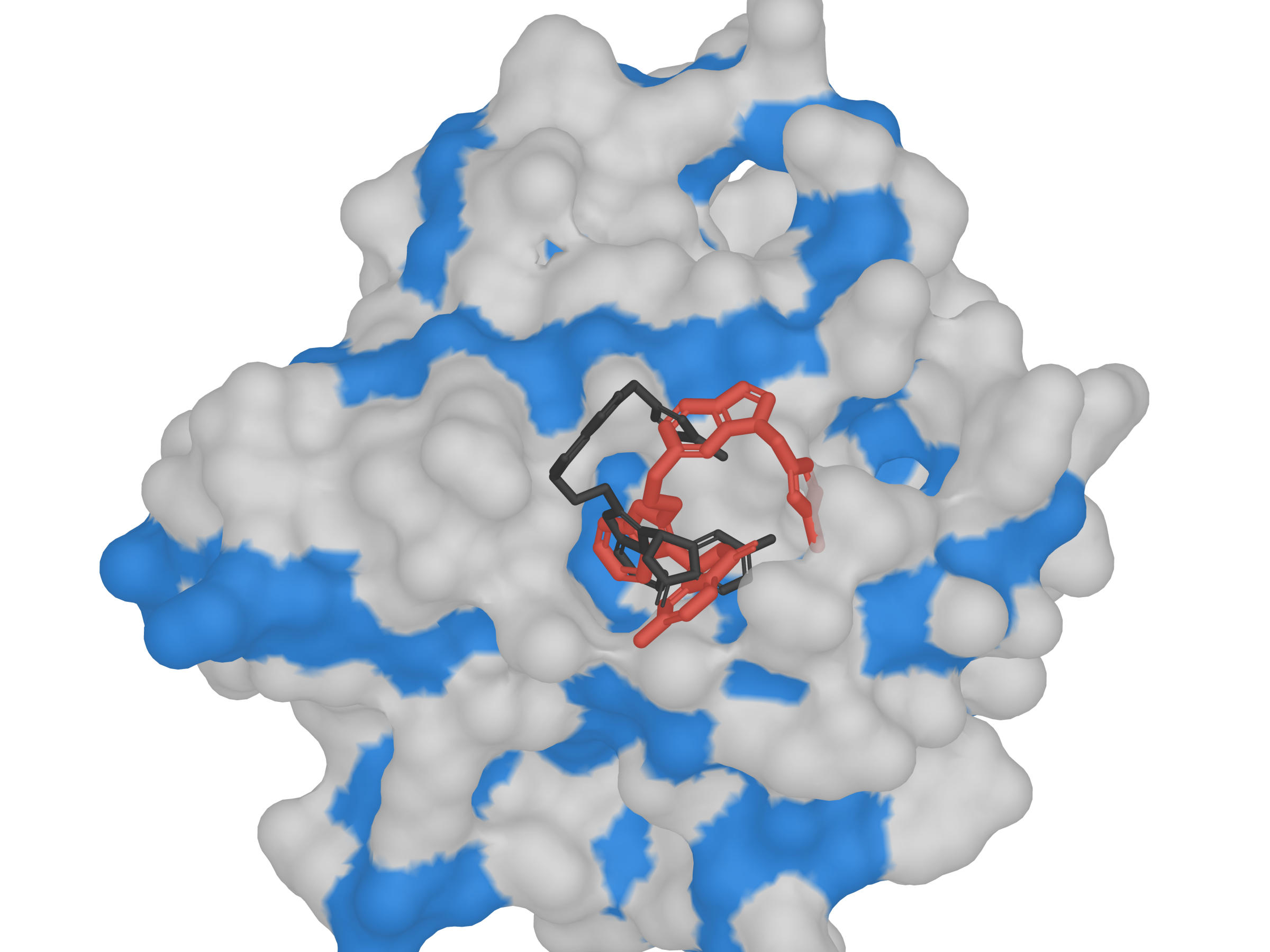} &
    \includegraphics[width=0.45\linewidth]{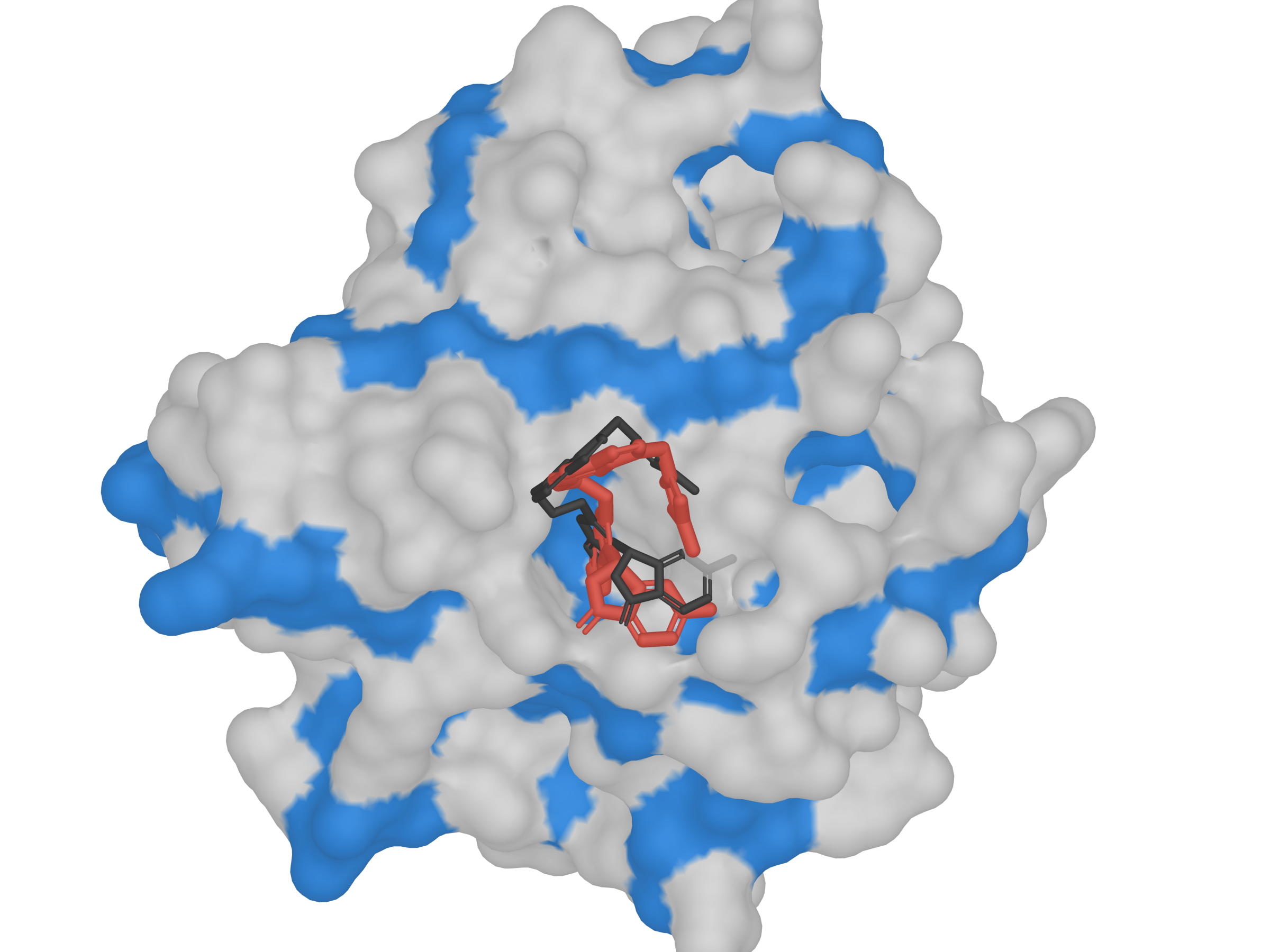} \\
    \small \textbf{6GJ8}: L-RMSD \textbf{5.05 \AA}, AA-RMSD \textbf{2.05 \AA} &
    \small \textbf{6GJ8}: L-RMSD \textbf{1.90 \AA}, AA-RMSD \textbf{1.95 \AA}
  \end{tabular}

  \caption{\small
    \textbf{Case studies on EBNA1 (6NPP) and KRAS G12D (6GJ8).} Ground truth ligand is \textbf{black} and predicted ligand is \textbf{{\color{red}red}}. Receptor is in its predicted configuration. In both examples, \textit{Harmony} produces a more accurate ligand pose and improved local all-atom reconstruction than the baseline.
  }
  \label{fig:case-studies}
  \vspace{-\intextsep} 
\end{figure}

\vspace{-0.5em}
\subsection{Ablation Studies} \label{subsec:ablations}

\vspace{-0.5em}
\paragraph{Ablation Study on Harmonic Modes.}
We vary the number of harmonic modes $K$ the model learns during training. If performance does not drop significantly from decreasing maximum $K$ it would indicate coarse dynamics of torsional movement from apo to holo structures. This is indeed the case: in \new{left} part of \Cref{fig:harmony-three-ablations}, we observe that \textit{Harmony} performs the best with \new{$K=8$}. Performance drops manageably for minimum $K=2$. The model starts to lose performance further when too many harmonic modes are used. That can be explained by the above idea of intrinsic low-frequency nature\new{: rotameric variation of molecular torsions is well captured by low-order modes, whereas high-frequency modes are weakly supervised, hard to estimate at moderate and high noise levels, and prone to overfitting}. For all of $K$, \textit{Harmony} outperforms latest rival FlexDock, which highlights the stability against hyperparameter choices.\looseness=-1

\vspace{-1.0em}
\paragraph{Ablation Study on Frequency Damping.}
We believe \textit{Harmony's} results are tied to the damping effect on torsional motion frequencies. In this ablation we remove the damping and allow model to learn all frequencies by itself. Results are reported in \Cref{fig:harmony-three-ablations} (middle) and conclude, that in such setting performance degrades and noising of high frequency components is the central part of the proposed model.

\vspace{-1.0em}
\paragraph{Ablation Study on Learning Signal Between Ligand and Side Chains.}
We would like to additionally answer, whether applying \textit{Harmony's} framework simultaneously to ligand and side-chain torsions improves their performance. To validate this, we trained two versions of our model: \textit{Harmony} only for ligand torsions and \textit{Harmony} only for side-chain torsions. As a result, ligand success rate drops without harmonic model for side chains and vice versa (see right side of \Cref{fig:harmony-three-ablations}). This supports the claim that representations learned by \textit{Harmony} benefit both ligand and protein positions.

\begin{newtext}
\vspace{-1.0em}
\paragraph{Ablation Study on Flexible Backbone.}
We also apply the harmonic torsional head to a Riemannian flow matching model with a flexible backbone, as used in FlexDock. The head improves both ligand and all-atom success rates, although more moderately than in the VE-SDE setting, since the mode-wise damping does not transfer to the flow matching formulation (Appendix~\ref{app:backbone}).
\end{newtext}

\subsection{Case Studies}

Beyond aggregate benchmarks, we examine two representative targets\new{, chosen because they combine therapeutic relevance with difficult pocket geometries, a broad polar site surrounded by a loop and a shallow, highly exposed site}: EBNA1 (6NPP) \citep{messick2019structure} and KRAS G12D (6GJ8) \citep{kessler2019drugging}.
For each complex, we compare the best confidence-ranked pose from the baseline model (FlexDock) and from \textit{Harmony} \new{(\Cref{fig:case-studies}). Additional \textit{Harmony} predictions are shown in \Cref{fig:visualizations-grid} (Appendix~\ref{app:vis})}. To verify that these targets are not unusually close to the training set, we additionally computed ligand, protein, and contact-based similarity scores to PDBBind training complexes as described in Appendix \ref{app:case-studies-similarity}. The mean combined similarity across the test set is \(0.196 \pm 0.019\), while 6GJ8 and 6NPP score \(0.189\) and \(0.183\), respectively, indicating that both case studies are less similar than a typical test target. We emphasize that these similarity scores are only a coarse check against obvious train-test leakage and do not directly measure pocket difficulty, which depends primarily on the spatial arrangement and chemistry of local residues. In all cases, docking is performed from an apo receptor conformation, making the setting substantially more challenging than rigid docking.

\textbf{EBNA1 (6NPP).}
EBNA1 is required for Epstein-Barr virus genome maintenance and is a therapeutically relevant viral target \citep{frappier2012epstein,jiang2018ebna1}. Geometrically, the ligand binds in a broad site with a surrounding loop. The baseline prediction fails to preserve this local closure, whereas \textit{Harmony} maintains the enclosing loop structure and improves both ligand and receptor pose.

\textbf{KRAS G12D (6GJ8).}
KRAS is a central signaling GTPase and one of the most important oncogenic targets in cancer \citep{prior2012comprehensive}. The binding site in 6GJ8 is shallow and highly exposed, so the ligand is only weakly constrained by shape complementarity. This makes the complex particularly difficult for docking from apo structure. In this regime, the baseline model produces a large ligand pose error, while \textit{Harmony} recovers a substantially more accurate conformation closer to the crystal reference.\looseness=-1

\vspace{-1.0em}
\section{Discussion} \label{sec:discussion}

\vspace{-1.0em}
\paragraph{Future Work and Limitations.} 
The main limitation is that \textit{Harmony} is based on a VE diffusion process that needs to have Gaussian noise prior. Thus, adapting this harmonic framework to a generative model for arbitrary endpoint distributions like flow-based interpolants is a compelling future direction as it would enable fully all-atom flexible docking with backbone motion. \new{Our model also keeps the backbone fixed, which limits targets with substantial backbone rearrangement.} Additionally, it would be interesting to test harmonic score formulation in other tasks requiring generation of geometric graphs such as \new{standalone molecular conformer generation,} protein or material design.

\vspace{-1.0em}
\paragraph{Conclusion.}
We introduced \textit{Harmony}, a harmonic torsional diffusion module for flexible protein-ligand docking. By parameterizing ligand and side-chain torsional scores in Fourier space and utilizing variance-exploding diffusion properties on the torus, \textit{Harmony} makes angular periodicity explicit and improves the modeling of flexible binding sites. Across standard benchmarks and case studies, this simple geometric change consistently improves pose recovery and all-atom reconstruction. More broadly, our results suggest that respecting the intrinsic manifold structure of torsional motion is not a minor modeling detail, but a key ingredient for accurate diffusion-based flexible docking.

\bibliographystyle{plainnat}
\bibliography{references}

\newpage
\appendix

\begin{newtext}
\section{Notation} \label{app:notation}

\Cref{tab:notation} summarizes the notation used throughout the paper.

\begin{table}[h]
\bluefloat
\centering
\caption{\small \new{Summary of notation.}}
\label{tab:notation}
\small
\renewcommand{\arraystretch}{1.15}
\begin{tabular}{p{0.30\linewidth}p{0.62\linewidth}}
\toprule
Symbol & Meaning \\
\midrule
\multicolumn{2}{l}{\textit{Complex representation}} \\
$\mathcal{G}$, $\mathcal{G}_{\mathrm{pocket}}$ & Heterogeneous complex graph and its pocket-centered subgraph \\
$\mathcal{V}^l$, $\mathcal{V}^b$, $\mathcal{V}^a$ & Ligand atoms, receptor residue nodes (at $C_\alpha$), receptor atoms \\
$\mathcal{E}^l$, $\mathcal{E}^b$, $\mathcal{E}^a$ & Intra-graph ligand, residue and receptor-atom edges \\
$\mathcal{E}^{lb}$, $\mathcal{E}^{la}$, $\mathcal{E}^{ab}$ & Ligand--residue, ligand--atom and atom--residue edges \\
\midrule
\multicolumn{2}{l}{\textit{Docking state}} \\
$\mathbf{z}=(\mathbf{t},\mathbf{R},\boldsymbol{\tau},\boldsymbol{\chi})$ & Ligand translation, ligand rotation, ligand torsions, side-chain torsions \\
$\mathcal{M}=\mathbb{R}^3\times SO(3)\times\mathbb{T}^{m_l}\times\mathbb{T}^{m_s}$ & Product manifold on which docking is modeled \\
$m_l$, $m_s$ & Number of rotatable ligand bonds and flexible side-chain torsions \\
$\mathbb{T}=\mathbb{R}/2\pi\mathbb{Z}$, $\mathbb{T}^m$ & Circle and $m$-dimensional torus \\
$\vartheta$ & A single torsion angle, i.e.\ one component of $\boldsymbol{\tau}$ or $\boldsymbol{\chi}$ \\
$e$ & Rotatable bond associated with torsion $\vartheta$ \\
$\mathrm{wrap}(\cdot)$ & Map of an angle to $(-\pi,\pi]$ \\
\midrule
\multicolumn{2}{l}{\textit{Diffusion process}} \\
$\lambda\in[0,1]$ & Diffusion time (noise level) \\
$g(\lambda)$, $\mathbf{w}_\lambda$ & Diffusion coefficient and standard Brownian motion of the VE-SDE \\
$\sigma(\lambda)$ & Noise scale, $\sigma(\lambda)=\sigma_{\min}^{1-\lambda}\sigma_{\max}^{\lambda}$ \\
$\sigma_u(\lambda)$, $u\in\{\mathrm{tr},\mathrm{rot},\mathrm{tor},\mathrm{sc}\}$ & Component-specific noise scale for translation, rotation, ligand torsions, side chains \\
$p_\lambda$ & Density of the noised state at time $\lambda$, and its perturbation kernel $p_\lambda(\cdot\mid\cdot)$ \\
$q_\lambda$ & Factorized forward noising distribution on $\mathcal{M}$ \\
$\boldsymbol{\epsilon}_u$ & Standard Gaussian noise for component $u$ \\
$\boldsymbol{\omega}_\lambda\sim\mathrm{IGSO}(3;\sigma_{\mathrm{rot}}(\lambda))$, $\mathrm{Exp}$ & Axis-angle rotation noise and the exponential map on $SO(3)$ \\
$\Delta$, $\Delta_{\mathbb{T}}$ & Laplacian and Laplace--Beltrami operator on $\mathbb{T}$ ($\Delta_{\mathbb{T}}=\partial_\vartheta^2$) \\
$\nu(\lambda)=g(\lambda)^2/2$ & Time-dependent diffusivity of the heat equation \\
\midrule
\multicolumn{2}{l}{\textit{Harmonic parameterization}} \\
$\mathbf{u}_e$ & Bond representation produced by the message-passing backbone \\
$\mathbf{g}_\theta$ & Head mapping $\mathbf{u}_e$ to harmonic coefficients \\
$K$, $k$ & Number of retained Fourier modes and mode index \\
$a_{e,k}$, $b_{e,k}$ & Cosine and sine coefficients of mode $k$ for bond $e$ \\
$\psi_e(\vartheta)$, $\psi_{e,\lambda}(\vartheta)$ & Harmonic potential and its heat-smoothed version at noise level $\lambda$ \\
$s_e(\vartheta,\lambda)=\partial_\vartheta\psi_{e,\lambda}(\vartheta)$ & Predicted torsional score \\
\bottomrule
\end{tabular}
\end{table}
\end{newtext}

\section{Proofs of Propositions}
\label{app:proofs}

\begin{proof}[Proof of Proposition~\ref{prop:ve-heat}]
For the driftless SDE
\[
d\mathbf{z}_\lambda = g(\lambda)\,d\mathbf{w}_\lambda,
\]
the infinitesimal generator is
\[
\mathcal{L}_\lambda = \frac{g(\lambda)^2}{2}\Delta.
\]
Since the Laplacian is self-adjoint, the Fokker--Planck equation is
\[
\partial_\lambda p_\lambda = \mathcal{L}_\lambda^\ast p_\lambda = \mathcal{L}_\lambda p_\lambda = \frac{g(\lambda)^2}{2}\Delta p_\lambda.
\]
Thus the forward density evolves by a heat equation with time-dependent diffusivity \(\nu(\lambda)=g(\lambda)^2/2\).
\end{proof}

\begin{proof}[Proof of Proposition~\ref{prop:fourier-eigenbasis}]
On the circle \(\mathbb{T}\) with angular coordinate \(\vartheta\), the Laplace--Beltrami operator is \(\Delta_{\mathbb{T}}=\partial_\vartheta^2\). Hence
\[
\Delta_{\mathbb{T}}\cos(k\vartheta)=\partial_\vartheta^2\cos(k\vartheta)=-k^2\cos(k\vartheta),
\qquad
\Delta_{\mathbb{T}}\sin(k\vartheta)=\partial_\vartheta^2\sin(k\vartheta)=-k^2\sin(k\vartheta),
\]
so \(\cos(k\vartheta)\) and \(\sin(k\vartheta)\) are eigenfunctions with eigenvalue \(-k^2\). Therefore, for
\[
\psi(\vartheta)=\sum_{k=1}^K \big[a_k\cos(k\vartheta)+b_k\sin(k\vartheta)\big],
\]
the heat semigroup acts modewise:
\begin{gather*}
\exp\!\left(\frac{\sigma^2}{2}\Delta_{\mathbb{T}}\right)\cos(k\vartheta)
=
\exp\!\left(-\frac12 k^2\sigma^2\right)\cos(k\vartheta),
\\
\exp\!\left(\frac{\sigma^2}{2}\Delta_{\mathbb{T}}\right)\sin(k\vartheta)
=
\exp\!\left(-\frac12 k^2\sigma^2\right)\sin(k\vartheta),
\end{gather*}
which yields
\[
\exp\!\left(\frac{\sigma^2}{2}\Delta_{\mathbb{T}}\right)\psi(\vartheta)
=
\sum_{k=1}^K
\exp\!\left(-\frac12 k^2\sigma^2\right)
\big[a_k\cos(k\vartheta)+b_k\sin(k\vartheta)\big].
\]
\end{proof}

\begin{proof}[Proof of Proposition~\ref{prop:wrapped-heat}]
Let
\[
\vartheta_\lambda=\mathrm{wrap}\!\left(\vartheta_0+\sigma(\lambda)\epsilon\right),
\qquad
\epsilon\sim\mathcal N(0,1).
\]
Before wrapping, \(\vartheta_0+\sigma(\lambda)\epsilon\) is Gaussian with mean \(\vartheta_0\) and variance \(\sigma^2(\lambda)\). Passing to the quotient \(\mathbb{R}/2\pi\mathbb{Z}\cong\mathbb{T}\) periodizes this Gaussian, giving the wrapped normal kernel
\[
p_\lambda(\vartheta\mid\vartheta_0)
=
\sum_{n\in\mathbb Z}
\frac{1}{\sqrt{2\pi\sigma^2(\lambda)}}
\exp\!\left(
-\frac{(\vartheta-\vartheta_0-2\pi n)^2}{2\sigma^2(\lambda)}
\right).
\]
The heat kernel on \(\mathbb{T}\) at time \(t\) is exactly the periodized Gaussian on \(\mathbb{R}\) with variance \(2t\), so setting \(t=\sigma^2(\lambda)/2\) gives
\[
p_\lambda(\cdot\mid\vartheta_0)
=
\exp\!\left(\frac{\sigma^2(\lambda)}{2}\Delta_{\mathbb{T}}\right)\delta_{\vartheta_0}.
\]
Differentiating with respect to \(\lambda\) yields
\[
\partial_\lambda p_\lambda
=
\frac{(\sigma^2)'(\lambda)}{2}\Delta_{\mathbb{T}}\,p_\lambda.
\]
In particular, if \(\sigma^2(\lambda)=\int_0^\lambda g(s)^2\,ds\), then \((\sigma^2)'(\lambda)=g(\lambda)^2\), and therefore
\[
\partial_\lambda p_\lambda
=
\frac{g(\lambda)^2}{2}\Delta_{\mathbb{T}}\,p_\lambda.
\]
\end{proof}

\section{Stratification of Validity Checks in PoseBusters}
\label{app:strat-posebusters}

In this section we provide a detailed breakdown of validity check passes in PoseBusters for baseline model (FlexDock) and \textit{Harmony} with no energy minimization. In \Cref{fig:posebusters-check-breakdown} it can be seen that our framework improves on protein-ligand distance, internal steric clashes and ligand strain.

\begin{figure*}[t]
  \centering
  \includegraphics[width=\textwidth]{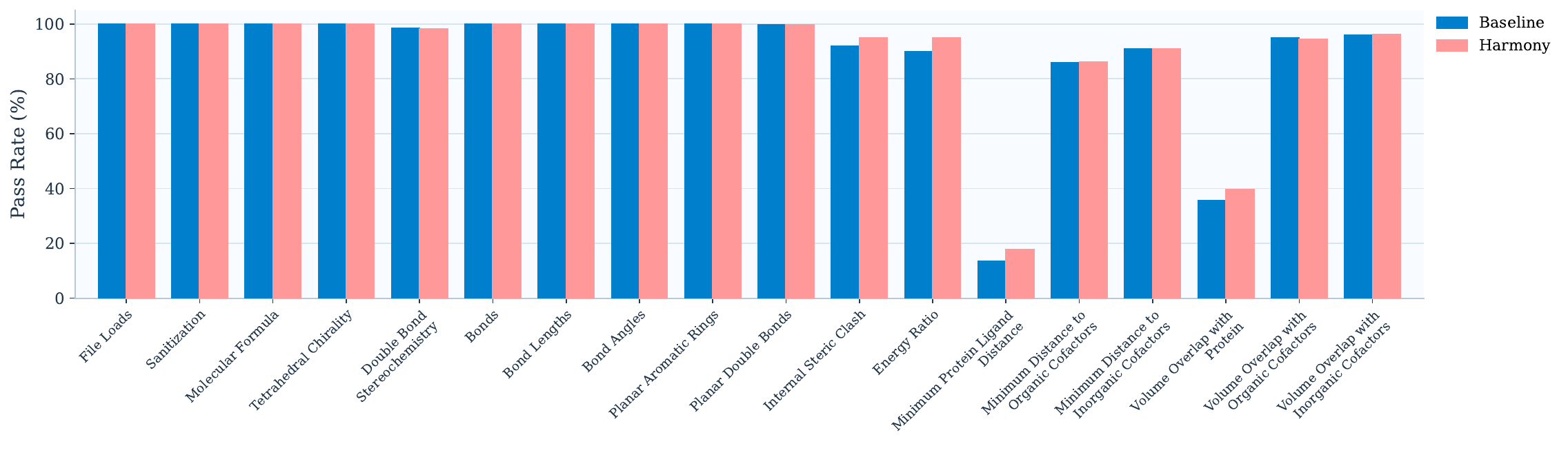}
  \caption{
    PoseBusters check-wise breakdown for the baseline model and \textit{Harmony}. Each bar reports the pass rate for one PoseBusters criterion, highlighting that \textit{Harmony} improves several physically meaningful checks, particularly those related to internal steric validity, ligand strain, and protein--ligand separation.
  }
  \label{fig:posebusters-check-breakdown}
\end{figure*}

\section{Training and Inference}

In this section we provide additional details on training and inference algorithms for our model, and data preparation such as pocket selection, conformer matching, alignment of protein frames and flexibility formalization.

\subsection{Algorithms for Training and Inference of \textit{Harmony}}

We provide comprehensive pseudocode for training and inference of \textit{Harmony} in Algorithm \ref{alg:harmony-training} and Algorithm \ref{alg:harmony-inference}.

\begin{algorithm}[t]
\caption{Training of \textit{Harmony}}
\label{alg:harmony-training}
\KwIn{Preprocessed training set $\mathcal{D}=\{(\mathcal{G}^{(i)}_{\mathrm{pocket}}, \mathbf{z}^{(i)}_0)\}_{i=1}^M$, where $\mathbf{z}_0=(\mathbf{t}_0,\mathbf{R}_0,\boldsymbol{\tau}_0,\boldsymbol{\chi}_0)$; noise schedules $\sigma_{\mathrm{tr}},\sigma_{\mathrm{rot}},\sigma_{\mathrm{tor}},\sigma_{\mathrm{sc}}$; model $f_\theta$}
\KwOut{Updated parameters $\theta$}

\ForEach{minibatch $\mathcal{B}\subset\mathcal{D}$}{
    Sample $\lambda \sim \mathcal{U}(0,1)$\;

    Sample translation noise $\boldsymbol{\epsilon}_{\mathrm{tr}} \sim \mathcal{N}(\mathbf{0},\mathbf{I}_3)$\;
    Sample rotational noise $\boldsymbol{\omega}_\lambda \sim \mathrm{IGSO}(3;\sigma_{\mathrm{rot}}(\lambda))$\;
    Sample torsional noises
    $\boldsymbol{\epsilon}_{\mathrm{tor}} \sim \mathcal{N}(\mathbf{0},\mathbf{I}_{m_l})$ and
    $\boldsymbol{\epsilon}_{\mathrm{sc}} \sim \mathcal{N}(\mathbf{0},\mathbf{I}_{m_s})$\;

    Construct the noised docking state
    \[
    \mathbf{t}_\lambda=\mathbf{t}_0+\sigma_{\mathrm{tr}}(\lambda)\boldsymbol{\epsilon}_{\mathrm{tr}},
    \qquad
    \mathbf{R}_\lambda=\mathrm{Exp}(\boldsymbol{\omega}_\lambda)\mathbf{R}_0,
    \]
    \[
    \boldsymbol{\tau}_\lambda=
    \mathrm{wrap}\!\left(\boldsymbol{\tau}_0+\sigma_{\mathrm{tor}}(\lambda)\boldsymbol{\epsilon}_{\mathrm{tor}}\right),
    \qquad
    \boldsymbol{\chi}_\lambda=
    \mathrm{wrap}\!\left(\boldsymbol{\chi}_0+\sigma_{\mathrm{sc}}(\lambda)\boldsymbol{\epsilon}_{\mathrm{sc}}\right).
    \]

    Apply $(\mathbf{t}_\lambda,\mathbf{R}_\lambda,\boldsymbol{\tau}_\lambda,\boldsymbol{\chi}_\lambda)$ to the clean pocket graph to obtain the noised graph $\mathcal{G}_\lambda$\;

    Predict rigid-body scores and harmonic coefficients
    \[
    \hat{\mathbf{s}}_{\mathrm{tr}},\;
    \hat{\mathbf{s}}_{\mathrm{rot}},\;
    \{\hat a^{\,l}_{e,k},\hat b^{\,l}_{e,k}\},\;
    \{\hat a^{\,sc}_{e,k},\hat b^{\,sc}_{e,k}\}
    \leftarrow
    f_\theta(\mathcal{G}_\lambda,\lambda).
    \]

    Evaluate ligand and side-chain torsional scores from the harmonic heads:
    \[
    \hat{\mathbf{s}}_{\mathrm{tor}}
    \leftarrow
    \partial_\vartheta \hat{\psi}^{\,l}_{\lambda}(\vartheta),
    \qquad
    \hat{\mathbf{s}}_{\mathrm{sc}}
    \leftarrow
    \partial_\vartheta \hat{\psi}^{\,sc}_{\lambda}(\vartheta),
    \]
    where each $\hat{\psi}_{e,\lambda}$ is the heat-damped harmonic potential from \Cref{eq:psi}\;

    Compute the score targets from the known forward kernel
    \[
    \mathbf{s}^{\star}_{\mathrm{tr}},\;
    \mathbf{s}^{\star}_{\mathrm{rot}},\;
    \mathbf{s}^{\star}_{\mathrm{tor}},\;
    \mathbf{s}^{\star}_{\mathrm{sc}}
    \leftarrow
    \nabla \log q_\lambda(\mathbf{z}_\lambda \mid \mathbf{z}_0).
    \]

    Compute the training loss
    \[
    \mathcal{L}
    =
    \mathcal{L}_{\mathrm{tr}}
    +
    \mathcal{L}_{\mathrm{rot}}
    +
    \mathcal{L}_{\mathrm{tor}}
    +
    \mathcal{L}_{\mathrm{sc}},
    \]
    with each term matching the corresponding predicted and target scores\;

    Update $\theta$ with Adam on $\mathcal{L}$\;
}
\end{algorithm}

\begin{algorithm}[t]
\caption{Inference with \textit{Harmony}}
\label{alg:harmony-inference}
\KwIn{Pocket-centered apo graph $\mathcal{G}^{\mathrm{apo}}_{\mathrm{pocket}}$, initial ligand conformer, number of reverse steps $N$, noise schedules $\sigma_{\mathrm{tr}},\sigma_{\mathrm{rot}},\sigma_{\mathrm{tor}},\sigma_{\mathrm{sc}}$, trained model $f_\theta$}
\KwOut{Predicted bound ligand pose and pocket side-chain arrangement}

Initialize
\[
\mathbf{t}_N \sim \mathcal{N}(\mathbf{0},\sigma_{\mathrm{tr},\max}^2\mathbf{I}_3),
\qquad
\mathbf{R}_N \sim \mathcal{U}(SO(3)),
\]
\[
\boldsymbol{\tau}_N \sim \mathcal{U}(\mathbb{T}^{m_l}),
\qquad
\boldsymbol{\chi}_N \sim \mathcal{U}(\mathbb{T}^{m_s}).
\]

Apply $(\mathbf{t}_N,\mathbf{R}_N,\boldsymbol{\tau}_N,\boldsymbol{\chi}_N)$ to the initial ligand and apo pocket side chains to obtain $\mathcal{G}_N$\;

\For{$n \leftarrow N$ \KwTo $1$}{
    Set $\lambda_n = n/N$ and $\lambda_{n-1} = (n-1)/N$\;
    Set $\Delta \lambda = \lambda_n - \lambda_{n-1}$\;

    Predict rigid-body scores and harmonic coefficients
    \[
    \hat{\mathbf{s}}_{\mathrm{tr}},\;
    \hat{\mathbf{s}}_{\mathrm{rot}},\;
    \{\hat a^{\,l}_{e,k},\hat b^{\,l}_{e,k}\},\;
    \{\hat a^{\,sc}_{e,k},\hat b^{\,sc}_{e,k}\}
    \leftarrow
    f_\theta(\mathcal{G}_n,\lambda_n).
    \]

    Evaluate torsional scores from the harmonic heads:
    \[
    \hat{\mathbf{s}}_{\mathrm{tor}}
    \leftarrow
    \partial_\vartheta \hat{\psi}^{\,l}_{\lambda_n}(\vartheta),
    \qquad
    \hat{\mathbf{s}}_{\mathrm{sc}}
    \leftarrow
    \partial_\vartheta \hat{\psi}^{\,sc}_{\lambda_n}(\vartheta).
    \]

    Sample reverse-time noises
    \[
    \boldsymbol{\xi}_{\mathrm{tr}} \sim \mathcal{N}(\mathbf{0},\mathbf{I}_3),\quad
    \boldsymbol{\xi}_{\mathrm{rot}} \sim \mathcal{N}(\mathbf{0},\mathbf{I}_3),
    \]
    \[
    \boldsymbol{\xi}_{\mathrm{tor}} \sim \mathcal{N}(\mathbf{0},\mathbf{I}_{m_l}),\quad
    \boldsymbol{\xi}_{\mathrm{sc}} \sim \mathcal{N}(\mathbf{0},\mathbf{I}_{m_s}).
    \]

    Perform one reverse VE update:
    \[
    \mathbf{t}_{n-1}
    =
    \mathbf{t}_n
    +
    g_{\mathrm{tr}}(\lambda_n)^2 \Delta\lambda \,\hat{\mathbf{s}}_{\mathrm{tr}}
    +
    g_{\mathrm{tr}}(\lambda_n)\sqrt{\Delta\lambda}\,\boldsymbol{\xi}_{\mathrm{tr}},
    \]
    \[
    \mathbf{R}_{n-1}
    =
    \mathrm{Exp}\!\left(
    g_{\mathrm{rot}}(\lambda_n)^2 \Delta\lambda \,\hat{\mathbf{s}}_{\mathrm{rot}}
    +
    g_{\mathrm{rot}}(\lambda_n)\sqrt{\Delta\lambda}\,\boldsymbol{\xi}_{\mathrm{rot}}
    \right)\mathbf{R}_n,
    \]
    \[
    \boldsymbol{\tau}_{n-1}
    =
    \mathrm{wrap}\!\left(
    \boldsymbol{\tau}_n
    +
    g_{\mathrm{tor}}(\lambda_n)^2 \Delta\lambda \,\hat{\mathbf{s}}_{\mathrm{tor}}
    +
    g_{\mathrm{tor}}(\lambda_n)\sqrt{\Delta\lambda}\,\boldsymbol{\xi}_{\mathrm{tor}}
    \right),
    \]
    \[
    \boldsymbol{\chi}_{n-1}
    =
    \mathrm{wrap}\!\left(
    \boldsymbol{\chi}_n
    +
    g_{\mathrm{sc}}(\lambda_n)^2 \Delta\lambda \,\hat{\mathbf{s}}_{\mathrm{sc}}
    +
    g_{\mathrm{sc}}(\lambda_n)\sqrt{\Delta\lambda}\,\boldsymbol{\xi}_{\mathrm{sc}}
    \right).
    \]

    Apply $(\mathbf{t}_{n-1},\mathbf{R}_{n-1},\boldsymbol{\tau}_{n-1},\boldsymbol{\chi}_{n-1})$ to obtain the updated graph $\mathcal{G}_{n-1}$\;
}

Return the final ligand coordinates and pocket side-chain coordinates induced by $(\mathbf{t}_0,\mathbf{R}_0,\boldsymbol{\tau}_0,\boldsymbol{\chi}_0)$\;
\end{algorithm}

\subsection{Data Preparation} \label{app:data-prep}

\begin{algorithm}[t]
\caption{Finding Flexible Bonds and Rotated Fragments}
\label{alg:flexible-bonds}
\KwIn{Directed bond graph $\mathcal{G}=(\mathcal{V},\mathcal{E})$, optional anchor atom $v_{\mathrm{base}}$}
\KwOut{Rotatable-edge mask $m_{\mathrm{rot}}$, fragment index $\mathcal{F}$}

Initialize $m_{\mathrm{rot}}\leftarrow \mathbf{0}$ and $\mathcal{F}\leftarrow \emptyset$\;

\ForEach{undirected bond \(\{u,v\}\) represented by two directed edges \((u,v)\) and \((v,u)\)}{
    Remove \(\{u,v\}\) from the undirected version of \(\mathcal{G}\)\;
    
    \If{the graph remains connected}{
        Mark both directions as non-rotatable\;
        continue\;
    }

    Let \(C_1,C_2\) be the two connected components\;

    \uIf{$v_{\mathrm{base}}$ is specified}{
        Choose as rotated fragment the component that does \emph{not} contain \(v_{\mathrm{base}}\)\;
    }
    \Else{
        Choose as rotated fragment the smaller component\;
    }

    Let \(C_{\mathrm{rot}}\) be the chosen fragment and let \(e_{\mathrm{rot}}\) be the directed edge whose orientation points into that fragment\;

    \If{$|C_{\mathrm{rot}}| > 1$}{
        Set \(m_{\mathrm{rot}}(e_{\mathrm{rot}})\leftarrow 1\)\;
    }

    Store the rotated atoms:
    \[
    \mathcal{F}\leftarrow \mathcal{F}\cup\{(e_{\mathrm{rot}},a): a\in C_{\mathrm{rot}}\}
    \]
}

Construct a dependency graph between rotatable edges from fragment overlap and topologically sort them\;
Reorder \(\mathcal{F}\) according to this torsional application order\;

\Return{$m_{\mathrm{rot}}, \mathcal{F}$}
\end{algorithm}

\begin{algorithm}[t]
\caption{Conformer Matching for Ligands and Side Chains}
\label{alg:conformer-matching}
\KwIn{Reference ligand pose $\mathbf{X}^{l,\star}$, apo receptor coordinates $\mathbf{X}^{a}$, holo receptor coordinates $\mathbf{X}^{h}$}
\KwOut{Matched ligand conformer $\tilde{\mathbf{X}}^{l}$, matched apo side-chain coordinates $\tilde{\mathbf{X}}^{a}$, torsional updates $\Delta\boldsymbol{\chi}$}

\textbf{Ligand conformer matching:}

Generate a fresh ligand conformer $\mathbf{X}^{l,0}$ with RDKit and force-field refinement\;
Identify ligand rotatable bonds $\mathcal{B}^l$\;
Optimize torsion angles $\boldsymbol{\tau}$ with differential evolution:
\[
\hat{\boldsymbol{\tau}}
=
\arg\min_{\boldsymbol{\tau}\in[-\pi,\pi]^{|\mathcal{B}^l|}}
\mathrm{RMSD}\!\left(\mathrm{Rotate}(\mathbf{X}^{l,0},\boldsymbol{\tau}),\mathbf{X}^{l,\star}\right)
\]
Set
\[
\tilde{\mathbf{X}}^{l}=\mathrm{Rotate}(\mathbf{X}^{l,0},\hat{\boldsymbol{\tau}})
\]

\vspace{0.3em}
\textbf{Side-chain conformer matching:}

Compute global apo--holo Kabsch alignment using receptor $C_\alpha$ atoms\;
Align apo side chains into local holo backbone frames using per-residue $N$--$C_\alpha$--$C$ frames\;
Initialize $\tilde{\mathbf{X}}^{a}\leftarrow \mathbf{X}^{a}_{\mathrm{aligned}}$\;

\ForEach{flexible residue $r$}{
    Extract residue-specific rotatable side-chain bonds $\mathcal{B}^a_r$\;
    \If{$\mathcal{B}^a_r=\emptyset$}{continue}
    Optimize residue torsions independently with differential evolution:
    \[
    \hat{\boldsymbol{\chi}}_r
    =
    \arg\min_{\boldsymbol{\chi}_r\in[-\pi,\pi]^{|\mathcal{B}^a_r|}}
    \mathrm{RMSD}\!\left(
    \mathrm{RotateResidue}(\tilde{\mathbf{X}}^{a},r,\boldsymbol{\chi}_r),
    \mathbf{X}^{h}
    \right)
    \]
    \If{the optimized residue RMSD improves}{
        Apply $\hat{\boldsymbol{\chi}}_r$ to $\tilde{\mathbf{X}}^{a}$\;
        Store the corresponding per-bond torsional updates in $\Delta\boldsymbol{\chi}$\;
    }
}

\Return{$\tilde{\mathbf{X}}^{l}, \tilde{\mathbf{X}}^{a}, \Delta\boldsymbol{\chi}$}
\end{algorithm}

Our data preparation largely follows DiffDock-Pocket \citep{plainer2023diffdock} and FlexDock \citep{corso2025composing}. Below are more details.

\paragraph{Pocket Selection.}
We use pocket-centered docking to reduce the search space and focus the model on the local receptor environment most relevant for binding. Starting from the all-atom apo receptor and the reference ligand pose, we first identify a set of nearby residues by computing the minimum distance from each receptor residue to the ligand and selecting residues within a pocket cutoff. In the implementation, this is done at the atom level and then aggregated to the residue level by taking the minimum atom--ligand distance within each residue. Let \(S_{\mathrm{near}}\) denote this initial set of nearby residues. We then define the pocket center as the mean of the corresponding apo \(C_\alpha\) coordinates,
\[
\mathbf{c}_{\mathrm{pocket}}
=
\frac{1}{|S_{\mathrm{near}}|}
\sum_{r\in S_{\mathrm{near}}}\mathbf{x}^{\mathrm{apo}}_{r,C_\alpha}.
\]
We additionally enlarge the crop by retaining all residues whose apo \(C_\alpha\) lies within a fixed buffer radius of \(\mathbf{c}_{\mathrm{pocket}}\) which is chosen to be 10 \AA{}. The buffer is needed to make model for robust to the choice of a pocket as noted in DiffDock-Pocket \citep{plainer2023diffdock} and FlexDock \citep{corso2025composing}.\new{ }This yields a pocket-centered residue subset together with the corresponding protein atoms. If the initial cutoff would produce too few residues, the nearest residues are added until a minimum pocket size is reached. After cropping all coordinates are recentered by subtracting \(\mathbf{c}_{\mathrm{pocket}}\). At inference time, the holo receptor coordinates are not available and we define pocket center for the apo receptor. 

\paragraph{Conformer Matching for Ligand and Side Chains.}
Crystal structures are used as targets, and their distance variables such as bond lengths or angles fundamentally differ from ones, produced by generated ligand conformers (e.g. in RDKit) or apo receptor poses derived from folding models (e.g ESMFold). Thus, usage of crystal coordinates would induce distributional shift during inference. As discussed in \citet{plainer2023diffdock} and \citet{corso2025composing}, to tackle this problem we need to perform conformer matching for ligand and side chains before training. Full algorithm can be examined in Algorithm \ref{alg:conformer-matching}.

For ligands, we generate a fresh conformer using RDKit with a force field refinement, identify the rotatable bonds, and then optimize the ligand torsion angles by differential evolution so that the generated conformer best matches the reference ligand pose. The resulting matched conformer is used as the input ligand, while the original ligand coordinates are retained as the reference pose. 

For receptor side chains, after apo-holo alignment, we optimize each flexible residue independently over its torsional degrees of freedom using differential evolution, with the objective of minimizing the RMSD between the moved side-chain atoms and the corresponding holo coordinates. The optimized torsional differences are stored per rotatable side-chain bond, and applying them to the aligned apo coordinates yields a conformer-matched apo structure whose side chains are as close as possible to the holo arrangement. This matched structure provides the target for our model.

\paragraph{Graph Flexibility.}
Flexibility is represented through directed torsional edges and binary masks indicating which bonds are actually rotatable. General algorithm is demonstrated in Algorithm \ref{alg:flexible-bonds}. For ligands, a covalent bond is considered rotatable if removing the corresponding undirected bond disconnects the molecular graph and the moved fragment contains more than one atom. Among the two directed versions of the bond, we retain the direction associated with the fragment that will be rotated and store the list of affected atoms for each torsion. This induces a mask over ligand bond edges together with a fragment index that specifies which atoms move when a given torsion is perturbed.

For receptor side chains, we construct a directed side-chain bond graph residue by residue, starting from standard PDB atom names and moving outward from \(C_\alpha\) along the side-chain topology. Glycine has no side-chain torsions and contributes no rotatable edges. Proline side-chain bonds are included in the graph but are marked non-rotatable. For all other residues, the same graph criterion is applied to identify side-chain bonds whose rotation moves a particular fragment.

\paragraph{Alignment of Apo-Holo Frames.}
The purpose of apo-holo alignment is to disentangle rigid-body differences, local backbone frame differences, and genuine side-chain conformational changes before defining the diffusion targets. The first step is a global Kabsch alignment of the apo receptor to the holo receptor using the \(C_\alpha\) atoms. This removes the trivial rigid-body offset between the two structures. The second step aligns apo side chains into the local backbone frames of the holo structure. Concretely, for each residue we construct a local frame from the \(N\)-\(C_\alpha\)-\(C\) backbone triplet, compute the rotation that maps the apo frame to the holo frame, and apply that rotation to the corresponding local side-chain coordinates. The resulting coordinates are stored as the aligned apo structure. At inference time, no holo structure is available, thus the apo structure itself defines the working frame without any alignment.

\section{Model Architecture}

\subsection{Graph Construction.} \label{app:graph}

Harmony uses a heterogeneous representation of the pocket complex with three node sets: ligand atoms, receptor residues at \(C_\alpha\) positions, and receptor atoms. Following FlexDock \citep{corso2025composing}, we distinguish between \emph{static} structural graphs constructed during preprocessing and \emph{dynamic} geometric graphs recomputed online from the current noised coordinates during message passing.

\paragraph{Static Structural Graphs.}
The ligand is first represented by its bidirected covalent bond graph extracted from the RDKit molecular graph. Each covalent bond contributes two directed edges, and the corresponding bond type is stored as a one-hot edge feature over the bond classes \(\{\text{single},\text{double},\text{triple},\text{aromatic},\text{dative}\}\). This covalent graph is also used to define ligand torsional degrees of freedom through the rotatable edge masks described in Appendix \ref{app:data-prep}.

For the receptor, we store one residue node per amino acid at the \(C_\alpha\) coordinate, and one atom node per protein atom in the selected pocket. We additionally store a deterministic atom-to-residue mapping edge connecting each receptor atom to its parent residue. Side-chain torsional topology is cached separately as a directed residue-wise bond graph constructed from standard PDB atom naming.

\paragraph{Dynamic Geometric Graphs.}
During training and inference, the equivariant backbone rebuilds its spatial neighborhoods from the current noised coordinates. The ligand atom graph is augmented with a geometric radius graph using a \(6\) \AA{} cutoff. Thus, ligand message passing is performed on the union of covalent edges and local geometric edges.

The residue-level receptor graph is constructed with a \(k\)-nearest-neighbor graph on \(C_\alpha\) coordinates, using at most \(24\) neighbors per residue, followed by a radius filter at \(15\) \AA{}. The receptor atom graph is built analogously with at most \(12\) neighbors per atom and a radius cutoff of \(6\) \AA{}.

We use three types of cross-graph interactions. First, ligand atoms are connected to receptor residue nodes by a radius graph with maximum distance \(80\) \AA{}. Second, ligand atoms are connected to receptor atoms within \(6\) \AA{}. Third, each receptor atom is connected to its parent residue through the static atom-to-residue assignment edges described above.

\subsection{Featurization} \label{app:featurization}

\paragraph{Ligand Features.}
Ligand atom features follow the standard RDKit-based featurization used in DiffDock-Pocket and FlexDock. For each ligand atom, we encode atomic number, chirality, degree, formal charge, implicit valence, total hydrogen count, number of radical electrons, hybridization state, aromaticity, number of rings, and binary membership indicators for rings of sizes \(3\) through \(8\). Covalent ligand edges carry a one-hot bond-type encoding as described above. For non-covalent geometric ligand edges added online by the radius graph, the bond-type channel is set to zero.

\paragraph{Receptor Features.}
Each receptor residue node is featurized by amino-acid identity, and, a precomputed ESM-2 embedding is concatenated to this categorical residue representation. Each receptor atom node is featurized by four categorical attributes: the identity of the parent residue, the atomic number, a coarse atom-name category, and a fine atom-name category. The coarse atom-name category is obtained from the first two characters of the PDB atom name (e.g., \texttt{CA}, \texttt{CB}, \texttt{ND}), while the fine category uses the full atom name (e.g., \texttt{CA}, \texttt{CD1}, \texttt{NE2}).

In addition to these categorical features, we store several auxiliary atom-level quantities used by the model, including van der Waals radii and boolean masks identifying \(C_\alpha\), \(C\), and \(N\) atoms. Each residue node is additionally augmented with the two local backbone vectors \( \mathbf{x}_N - \mathbf{x}_{C_\alpha}\) and \(\mathbf{x}_C - \mathbf{x}_{C_\alpha}\), yielding a \(6\)-dimensional continuous orientation feature.

\paragraph{Geometric and Diffusion Features.}
Beyond the static node features above, Harmony augments all node types with diffusion-time embeddings computed from the current noise level. We use sinusoidal time embeddings of dimension \(32\). Edge distances are expanded with radial basis functions of dimension \(32\) for both intra-graph and cross-graph interactions. The equivariant backbone further computes spherical harmonics from the relative edge vectors, so each geometric edge is represented by its source-node time embedding, distance expansion, and directional information derived from the relative coordinates.

\paragraph{Torsional Metadata.}
Finally, preprocessing stores the torsional metadata required for ligand and side-chain updates: rotatable edge masks, fragment indices specifying which atoms move under each torsion, and matched torsional targets obtained from conformer matching (Appendix \ref{app:data-prep}).

\subsection{Model} \label{app:model}

\paragraph{Architecture Details.}
\textit{Harmony} is instantiated as a heterogeneous e3nn \citep{geiger2022e3nn} $SE(3)$-equivariant network over three coupled graphs: the ligand atom graph, the residue-level receptor backbone graph defined on $C_\alpha$ atoms, and the receptor atom graph in the selected pocket. The shared backbone uses $6$ equivariant tensor-product convolution layers with hidden widths $n_s=60$ for scalar channels and $n_v=15$ for higher-order channels, and spherical harmonics up to degree $l_{\max}=2$. We use layer normalization, SiLU activations, and dropout with rate $0.1$ throughout.

Node features are first embedded from categorical atom-level and residue-level descriptors together with sinusoidal diffusion-time embeddings of dimension $32$. Pairwise distances are expanded with RBFs of dimension $32$ for both intra-graph and cross-graph edges. The equivariant message-passing neighborhoods are constructed separately at each scale, with ligand, residue, atom, and cross-graph interactions handled by the shared heterogeneous backbone. Full graph construction and featurization details are deferred to \Cref{app:graph}.

On top of the shared equivariant trunk, \textit{Harmony} uses three prediction modules. A rigid-body head predicts ligand translation and rotation scores. A ligand torsion head predicts harmonic coefficients for ligand torsional score fields, and a side-chain torsion head predicts harmonic coefficients for flexible pocket side-chain torsions. The side-chain torsion head additionally conditions the harmonic coefficients on the diffusion noise level. These predicted coefficients are converted into periodic torsional scores through the harmonic parameterization introduced in \Cref{subsec:harmony}.

\paragraph{Confidence Head.}
Rather than training a separate confidence model, we attach a lightweight confidence head to the shared docking encoder. This head takes the pooled complex representation produced by the main docking network and predicts a scalar confidence score for each sampled pose, so confidence estimation is learned jointly with docking while adding only minimal computational overhead.

As the supervision target, we use protein--ligand interface lDDT (PLI-lDDT) \citep{lu2024dynamicbind}, which measures how well the predicted complex preserves native receptor--ligand contact distances. Specifically, PLI-lDDT averages, over atom pairs that are in contact in the reference complex, the fraction of predicted distances that fall within several tolerance thresholds. This makes it a natural confidence target, since it reflects local binding-site quality rather than only global coordinate error.

To focus confidence learning on geometries that are close to the final denoised samples, we use a time-weighted regression loss. If $\hat{c}_n$ denotes the predicted confidence for sample $n$, $c_n$ the corresponding PLI-lDDT target, and $\lambda_n$ its diffusion level, we optimize
\[
\mathcal{L}_{\mathrm{conf}}
=
\frac{\sum_n w(\lambda_n)\,|\hat{c}_n-c_n|}{\sum_n w(\lambda_n)},
\qquad
w(\lambda)=(1-\lambda)^2.
\]
This weighting emphasizes low-noise states, which are more relevant to the final docking prediction, while reducing the influence of highly corrupted intermediate samples.

\paragraph{OpenEquivariance CUDA Kernels for Tensor Products.}

Our implementation uses OpenEquivariance \citep{bharadwaj2025efficient} as the backend for the Clebsch--Gordan tensor products inside the e3nn equivariant convolutions. These tensor products are the main low-level bottleneck of $SE(3)$-equivariant message passing, since they are evaluated repeatedly across ligand, residue, and atom graphs at every layer. OpenEquivariance replaces generic tensor-product execution with JIT-compiled CUDA kernels specialized to the irreducible representation layout of the model. In particular, the generated kernels exploit the structured sparsity of Clebsch-Gordan coefficient tensors, precompute efficient execution schedules from the fixed irrep structure, and reduce memory consumption by staging small contractions. In \textit{Harmony}, enabling these kernels leaves the architecture and training objective unchanged, but substantially improves throughput of the equivariant backbone. Empirically, this reduces end-to-end training time from roughly $4.5$ days to $2$ days on $8$ H100 80GB GPUs.

\section{Experimental Details} \label{app:exp-details}

\subsection{Data} \label{app:data}

We train \textit{Harmony} on PDBBind \citep{liu2017forging}, using essentially the same data pipeline \new{as FlexDock} \citep{corso2025composing}. In particular, we adopt the standard time-based split, where complexes deposited before 2019 are divided into training and validation sets, while 363 complexes deposited after 2019 are reserved for testing.

To construct paired holo and apo receptor structures, the protein-ligand complexes in PDBBind are processed \new{with} PDBFixer from the OpenMM \citep{eastman2023openmm} to standardize non-canonical residues and complete missing heavy atoms. These corrected complex structures define our holo references. We then extract the protein sequence from the processed receptor and predict an apo structure with ESMFold \citep{lin2022language}. The resulting repaired ESMFold structures are used as apo receptors. Hydrogen atoms are removed from both apo and holo structures.

In addition to pose-accuracy metrics, we evaluate the physical plausibility of generated complexes with PoseBusters \citep{buttenschoen2024posebusters}. This benchmark assesses whether predicted poses satisfy basic chemical and geometric constraints, including ligand internal validity and protein-ligand interaction checks such as steric clash detection.

\subsection{Evaluation Metrics} \label{app:eval-metrics}

We evaluate docking quality using both ligand-level and complex-level geometric criteria. Our primary ligand metric is the ligand root-mean-square deviation (L-RMSD) between the predicted ligand pose and the reference crystal pose\new{, computed after aligning the predicted pocket to the reference structure and without superimposing the ligand}. For a ligand with \(N_l\) atoms, this is
\[
\mathrm{L\mbox{-}RMSD}
=
\sqrt{
\frac{1}{N_l}
\sum_{i=1}^{N_l}
\left\|
\hat{\mathbf{x}}^{\,l}_i-\mathbf{x}^{\,l,\star}_i
\right\|_2^2
},
\]
where \(\hat{\mathbf{x}}^{\,l}_i\) and \(\mathbf{x}^{\,l,\star}_i\) denote the predicted and reference coordinates of ligand atom \(i\), respectively. Following standard docking practice, we report the fraction of complexes satisfying
\[
\mathrm{L\mbox{-}RMSD} < 2 \text{\AA},
\]
which we refer to as the \textit{ligand success rate}.

To assess whether the full predicted complex is geometrically accurate, including the flexible binding-site side chains, we also measure all-atom RMSD (AA-RMSD) over the receptor atoms retained in the pocket representation. Let \(N_a\) denote the total number of evaluated atoms. Then
\[
\mathrm{AA\mbox{-}RMSD}
=
\sqrt{
\frac{1}{N_a}
\sum_{i=1}^{N_a}
\left\|
\hat{\mathbf{x}}_i-\mathbf{x}^{\star}_i
\right\|_2^2
}.
\]
We report the fraction of predictions with
\[
\mathrm{AA\mbox{-}RMSD} < 1 \text{\AA},
\]
which gives a stricter measure of local all-atom reconstruction quality.

\subsection{Similarity Analysis for Case Studies} \label{app:case-studies-similarity}

To contextualize the qualitative case studies, we compared each target complex against the PDBBind training split using a simple descriptor-based similarity analysis.

\paragraph{Ligand Similarity.}
For each ligand, we compute a Morgan fingerprint \citep{morgan1965generation} with radius \(2\) and \(1024\) bits after removing hydrogens. Let \(F_l(A)\) and \(F_l(B)\) denote the sets of active fingerprint bits for two ligands \(A\) and \(B\). Ligand similarity is measured by the Dice coefficient
\[
S_{\mathrm{lig}}(A,B)
=
\frac{2\,|F_l(A)\cap F_l(B)|}{|F_l(A)|+|F_l(B)|}.
\]

\paragraph{Protein Similarity.}
For each receptor, we parse the protein structure file and convert the amino-acid sequence into a one-letter sequence. Protein similarity is then computed as global sequence identity obtained from a Needleman--Wunsch alignment \citep{needleman1970general}. If \(L_{\mathrm{match}}(A,B)\) denotes the number of matched residues in the optimal global alignment of proteins \(A\) and \(B\), and \(L_{\mathrm{align}}(A,B)\) is the aligned length including gaps, then
\[
S_{\mathrm{prot}}(A,B)
=
\frac{L_{\mathrm{match}}(A,B)}{L_{\mathrm{align}}(A,B)}.
\]

\paragraph{Contact Interaction Similarity.}
To obtain a coarse description of the local binding environment without relying on explicit interaction typing, we build a residue-contact signature from heavy-atom distances between the ligand and the receptor. For each receptor residue, we compute the minimum heavy-atom distance to the ligand and record the residue if that distance is below \(4.5\) \AA{}. Contacts are further divided into two distance buckets: \emph{close} for distances at most \(3.5\) \AA{} and \emph{near} for distances between \(3.5\) and \(4.5\) \AA{}. Repeated occurrences of the same residue type within a bucket are counted separately, yielding a contact signature \(F_c\). We then compare two complexes with the Jaccard similarity
\[
S_{\mathrm{cont}}(A,B)
=
\frac{|F_c(A)\cap F_c(B)|}{|F_c(A)\cup F_c(B)|}.
\]

\paragraph{Combined Similarity Score.}
For each target complex, we compare its descriptor triplet against all complexes in the training split and average the three similarity components:
\[
S_{\mathrm{comb}}(A,B)
=
\frac{
S_{\mathrm{lig}}(A,B)
+
S_{\mathrm{prot}}(A,B)
+
S_{\mathrm{cont}}(A,B)
}{3}.
\]

\begin{newtext}
\section{Additional Results} \label{app:additional}

\subsection{Evaluation Protocol and Top-$k$ Results} \label{app:extended-baselines}

For \textit{Harmony}, DiffDock-Pocket, FIGRDock and DiffBindFR in \Cref{tab:docking_results}, we train three models with different random seeds and report the mean and standard deviation of each metric across these three models. Results for ReDock, FlexDock and the classical search-based methods are taken from prior work.

Besides the confidence-ranked top-1 prediction, we evaluate two selection rules that \rev{rely less on} the confidence model. \emph{Top-5} takes the lowest L-RMSD among the five highest-confidence samples\rev{, so the confidence model only needs to rank a good pose among its top five}. \emph{Oracle-10} takes the lowest L-RMSD among all ten generated samples \rev{and does not use the confidence model at all}. These metrics \rev{help} separate the quality of the generated poses from the accuracy of the ranking model. As shown in \Cref{tab:topk}, \textit{Harmony} improves over FlexDock under all three selection rules.

\begin{table}[h]
\bluefloat
\centering
\caption{\small \new{Ligand success rate (L-RMSD $<2$\AA{}, \%) on PDBBind ESMFold under different sample selection rules (mean\,$\pm$\,std over three training seeds).}}
\label{tab:topk}
\small
\begin{tabular}{lccc}
\toprule
Method & Top-1 $\uparrow$ & Top-5 $\uparrow$ & Oracle-10 $\uparrow$ \\
\midrule
FlexDock & 39.7\,$\pm$\,1.1 & 47.7\,$\pm$\,0.6 & 53.9\,$\pm$\,1.2 \\
\textit{Harmony} & \textbf{64.2\,$\pm$\,0.7} & \textbf{70.2\,$\pm$\,0.3} & \textbf{73.5\,$\pm$\,0.8} \\
\bottomrule
\end{tabular}
\end{table}

\subsection{Isolating the Torsional Parameterization} \label{app:isolation}

Although \textit{Harmony} follows the general architecture of DiffDock-Pocket and FlexDock, the models differ in some training parameters and in the confidence head. To isolate the effect of the torsional parameterization, we implement the \textit{Harmony} torsional head directly within DiffDock-Pocket, replacing only its ligand and side-chain torsional score heads and keeping all remaining components unchanged. As shown in \Cref{tab:isolation}, this single change substantially improves both ligand and pocket accuracy, indicating that the performance gains come from the harmonic parameterization itself rather than from architectural differences.

\begin{table}[h]
\bluefloat
\centering
\caption{\small \new{Replacing only the torsional head of DiffDock-Pocket with the \textit{Harmony} head on PDBBind ESMFold (top-1).}}
\label{tab:isolation}
\small
\begin{tabular}{lcc}
\toprule
Model & L-RMSD $<2$\AA{} (\%) $\uparrow$ & AA-RMSD $<1$\AA{} (\%) $\uparrow$ \\
\midrule
DiffDock-Pocket & 41.8 & 32.4 \\
DiffDock-Pocket + \textit{Harmony} head & \textbf{59.9} & \textbf{45.2} \\
\bottomrule
\end{tabular}
\end{table}

\subsection{Torsion-Specific Metrics} \label{app:torsion-metrics}

To verify that \textit{Harmony} improves the angular variables it directly models, and not only Cartesian reconstruction metrics, we evaluate torsion-level accuracy of the confidence-ranked top-1 predictions on the PoseBusters benchmark.

\paragraph{Ligand torsions.} \emph{Circular MAE} is the mean shortest angular distance between predicted and reference torsions, accounting for periodicity (e.g., the error between $179^\circ$ and $-179^\circ$ is $2^\circ$ rather than $358^\circ$). \emph{Accuracy@15$^\circ$} and \emph{Accuracy@30$^\circ$} are the fractions of modeled ligand torsions whose circular error is at most $15^\circ$ or $30^\circ$. \emph{All-ligand success@30$^\circ$} is the fraction of complexes for which every modeled ligand torsion is within $30^\circ$ of the reference. This is stricter than per-torsion accuracy, since a single incorrect dihedral makes the entire ligand unsuccessful.

\paragraph{Side-chain torsions.} \emph{$\chi_1$ and $\chi_2$ MAE} are the mean circular errors of the first and second side-chain dihedral angles. \emph{$\chi_1$ and $\chi_2$ accuracy@20$^\circ$} are the fractions of the corresponding angles predicted within $20^\circ$ of the reference. \emph{Rotamer recovery@20$^\circ$} is the fraction of evaluated residues for which all modeled $\chi$ angles are within $20^\circ$ of the reference, measuring recovery of the complete side-chain conformation.

As shown in \Cref{tab:torsion-ligand,tab:torsion-sidechain}, \textit{Harmony} improves all ligand and side-chain torsion metrics over FlexDock, demonstrating that the gains extend to the angular variables directly targeted by the harmonic parameterization.

\begin{table}[h]
\bluefloat
\centering
\caption{\small \new{Ligand torsion accuracy on PoseBusters (top-1).}}
\label{tab:torsion-ligand}
\small
\begin{tabular}{lcccc}
\toprule
Method & Circular MAE ($^\circ$) $\downarrow$ & Acc@15$^\circ$ (\%) $\uparrow$ & Acc@30$^\circ$ (\%) $\uparrow$ & All-ligand@30$^\circ$ (\%) $\uparrow$ \\
\midrule
FlexDock & 57.7 & 28.0 & 45.0 & 10.0 \\
\textit{Harmony} & \textbf{50.2} & \textbf{33.0} & \textbf{49.0} & \textbf{16.0} \\
\bottomrule
\end{tabular}
\end{table}

\begin{table}[h]
\bluefloat
\centering
\caption{\small \new{Side-chain torsion accuracy on PoseBusters (top-1).}}
\label{tab:torsion-sidechain}
\small
\begin{tabular}{lccccc}
\toprule
Method & $\chi_1$ MAE ($^\circ$) $\downarrow$ & $\chi_2$ MAE ($^\circ$) $\downarrow$ & $\chi_1$@20$^\circ$ (\%) $\uparrow$ & $\chi_2$@20$^\circ$ (\%) $\uparrow$ & Rotamer@20$^\circ$ (\%) $\uparrow$ \\
\midrule
FlexDock & 23.7 & 23.0 & 76.4 & 70.0 & 64.0 \\
\textit{Harmony} & \textbf{14.6} & \textbf{18.8} & \textbf{89.0} & \textbf{82.0} & \textbf{73.0} \\
\bottomrule
\end{tabular}
\end{table}

\subsection{OpenMM Relaxation} \label{app:relaxation}

Evaluating unrelaxed structures does not show how methods behave after the standard physical minimization commonly used in docking pipelines. We therefore apply the same OpenMM \citep{eastman2023openmm} relaxation protocol to the top-1 complexes of FlexDock and \textit{Harmony} on the PoseBusters benchmark set \citep{buttenschoen2024posebusters}. \Cref{tab:relaxation} reports the fraction of complexes with L-RMSD $<2$\AA{} that also pass all PoseBusters checks. After identical relaxation, \textit{Harmony} reaches a success rate of 44.2\% compared with 34.2\% for FlexDock.

\begin{table}[h]
\bluefloat
\centering
\caption{\small \new{Success rate (L-RMSD $<2$\AA{} and PB-valid) on the PoseBusters benchmark set after OpenMM relaxation.}}
\label{tab:relaxation}
\small
\begin{tabular}{lc}
\toprule
Method & L-RMSD $<2$\AA{} and PB-valid (\%) $\uparrow$ \\
\midrule
FlexDock + OpenMM & 34.2 \\
\textit{Harmony} + OpenMM & \textbf{44.2} \\
\bottomrule
\end{tabular}
\end{table}

\subsection{Flexible Backbone with Riemannian Flow Matching} \label{app:backbone}

Keeping the backbone fixed is a limitation of the current model, but not a conceptual restriction of \textit{Harmony}: the harmonic parameterization applies to any periodic degree of freedom, including backbone torsions. However, the mode-wise damping used by \textit{Harmony} follows directly from VE diffusion on the circle. A model with full backbone flexibility generally requires Riemannian flow matching with arbitrary endpoint distributions, where the same closed-form VE interpretation does not transfer unchanged, so we expect a smaller improvement in this setting. To test this, we take the Riemannian flow matching docking model with flexible backbone used in FlexDock \citep{corso2025composing} and replace its torsional heads with the harmonic torsional head. As shown in \Cref{tab:backbone}, the harmonic head improves both metrics, although more moderately than in the VE-SDE setting. Since side-chain motion captures many local pocket adaptations, this motivates our present scope, and developing Fourier parameterizations tailored to flow models with arbitrary endpoints is a promising direction for future work.

\begin{table}[h]
\bluefloat
\centering
\caption{\small \new{Harmonic torsional head in a Riemannian flow matching model with flexible backbone on PDBBind ESMFold (top-1).}}
\label{tab:backbone}
\small
\begin{tabular}{lcc}
\toprule
Full-backbone model & L-RMSD $<2$\AA{} (\%) $\uparrow$ & AA-RMSD $<1$\AA{} (\%) $\uparrow$ \\
\midrule
Riemannian flow matching baseline & 38.9 & 39.9 \\
+ harmonic torsional head & \textbf{41.5} & \textbf{41.0} \\
\bottomrule
\end{tabular}
\end{table}

\end{newtext}

\newpage
\section{Visualizations} \label{app:vis}

\begin{figure}[H]
  \centering

  \begin{tabular}{cc}
    \includegraphics[width=0.48\textwidth]{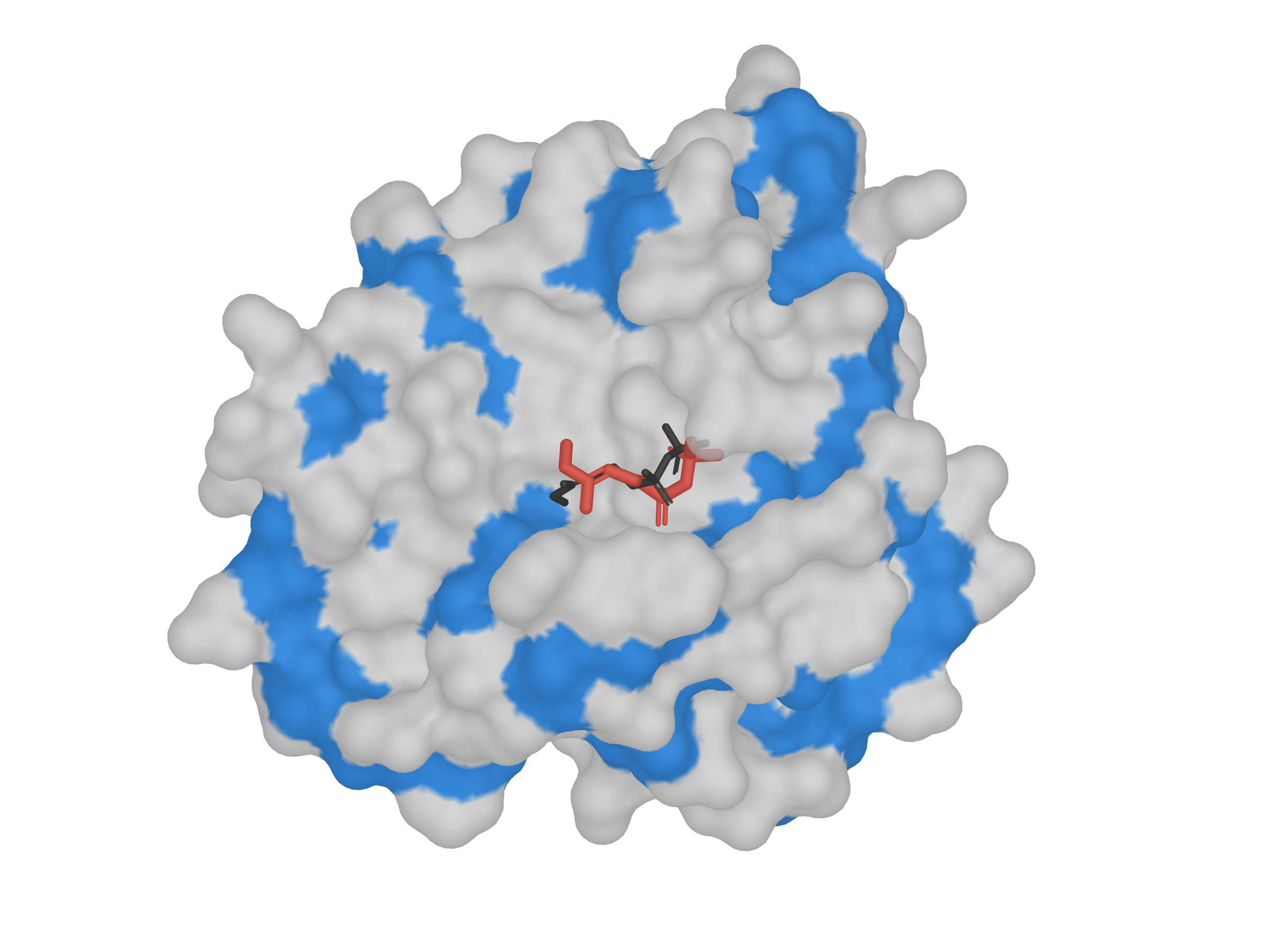} &
    \includegraphics[width=0.48\textwidth]{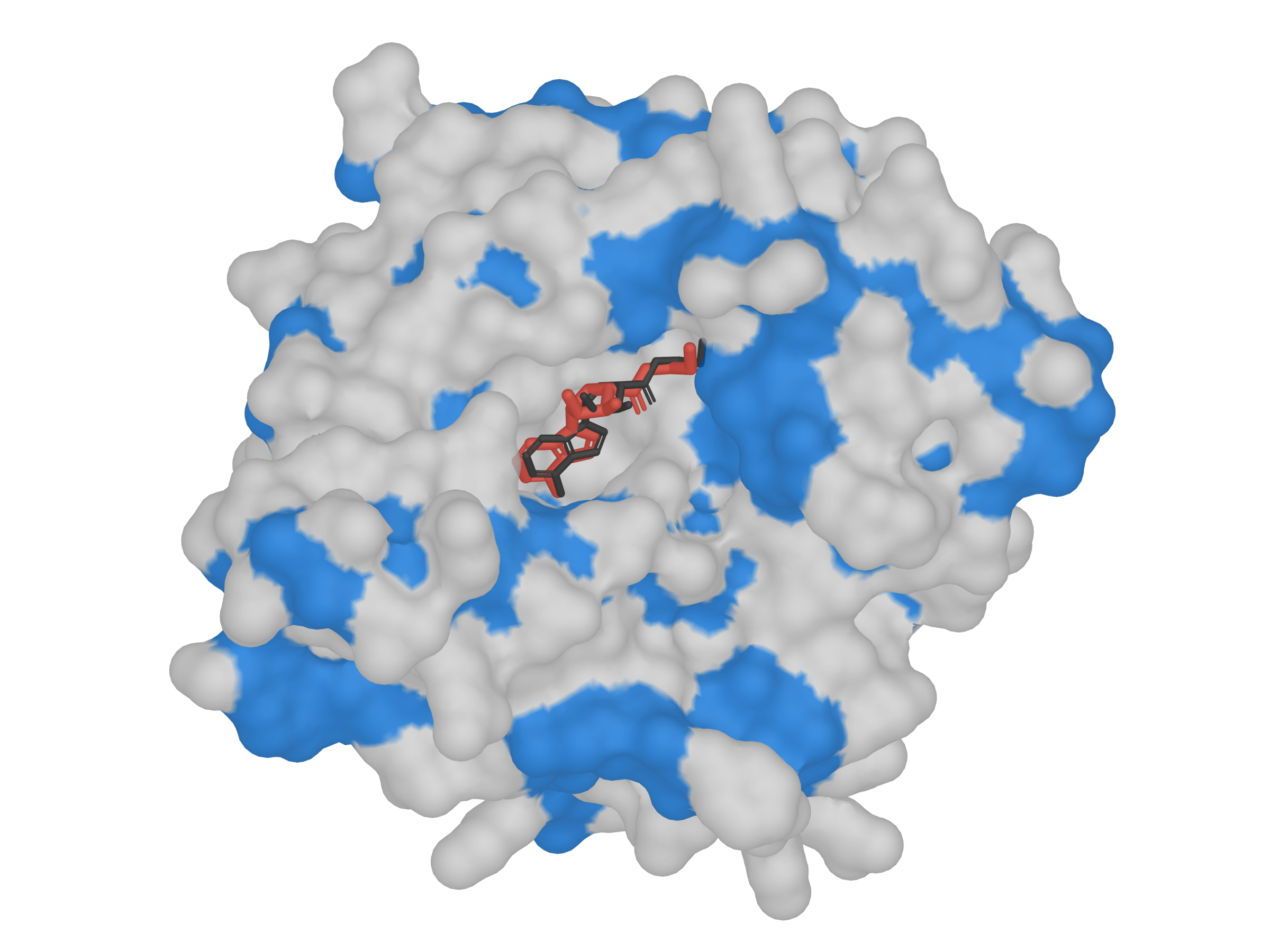} \\
    \small \textbf{5ZXK}: L-RMSD \textbf{2.53 \AA}, AA-RMSD \textbf{1.19 \AA} &
    \small \textbf{6CYG}: L-RMSD \textbf{0.53 \AA}, AA-RMSD \textbf{1.47 \AA} \\[1.0em]

    \includegraphics[width=0.48\textwidth]{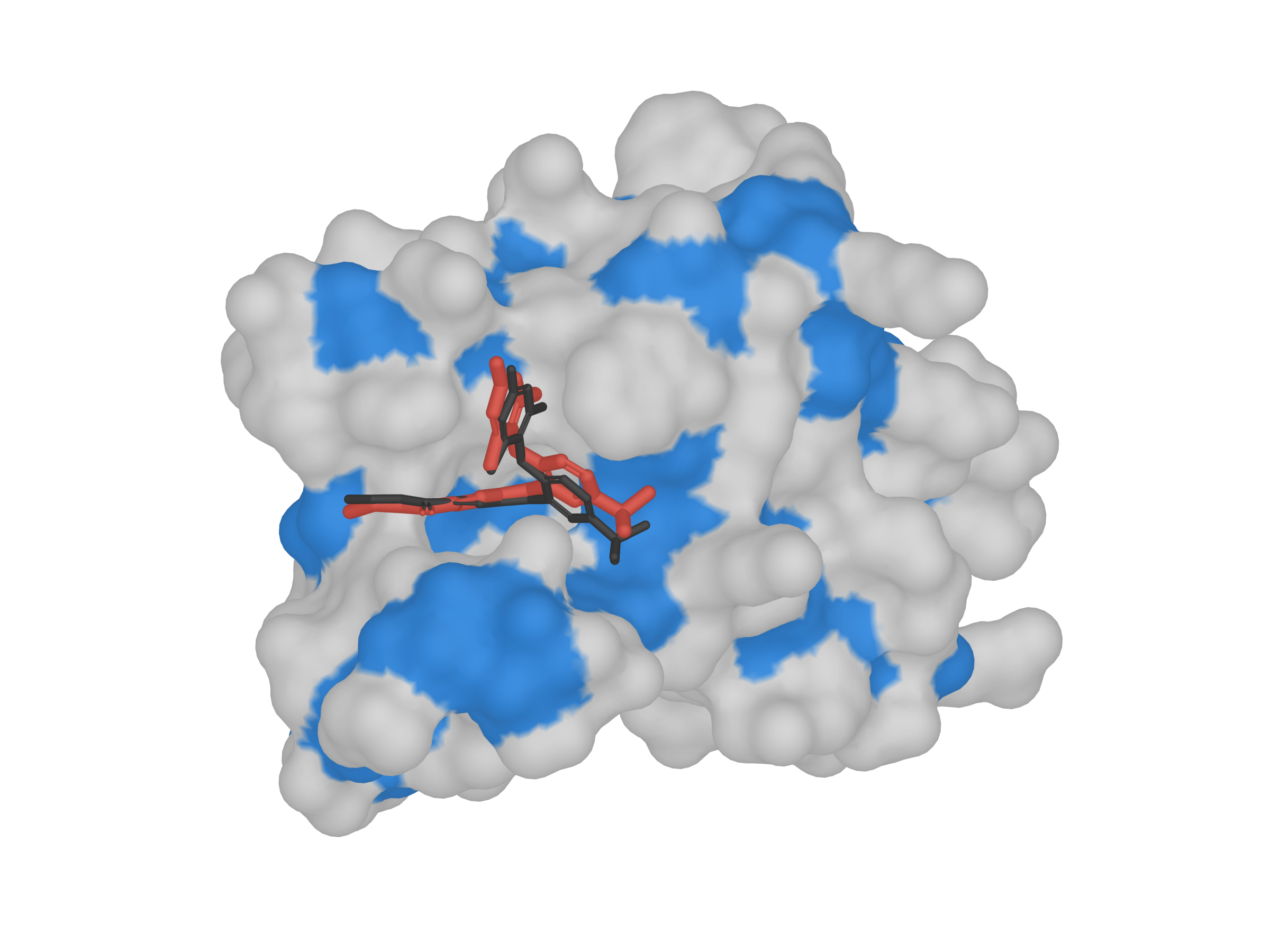} &
    \includegraphics[width=0.48\textwidth]{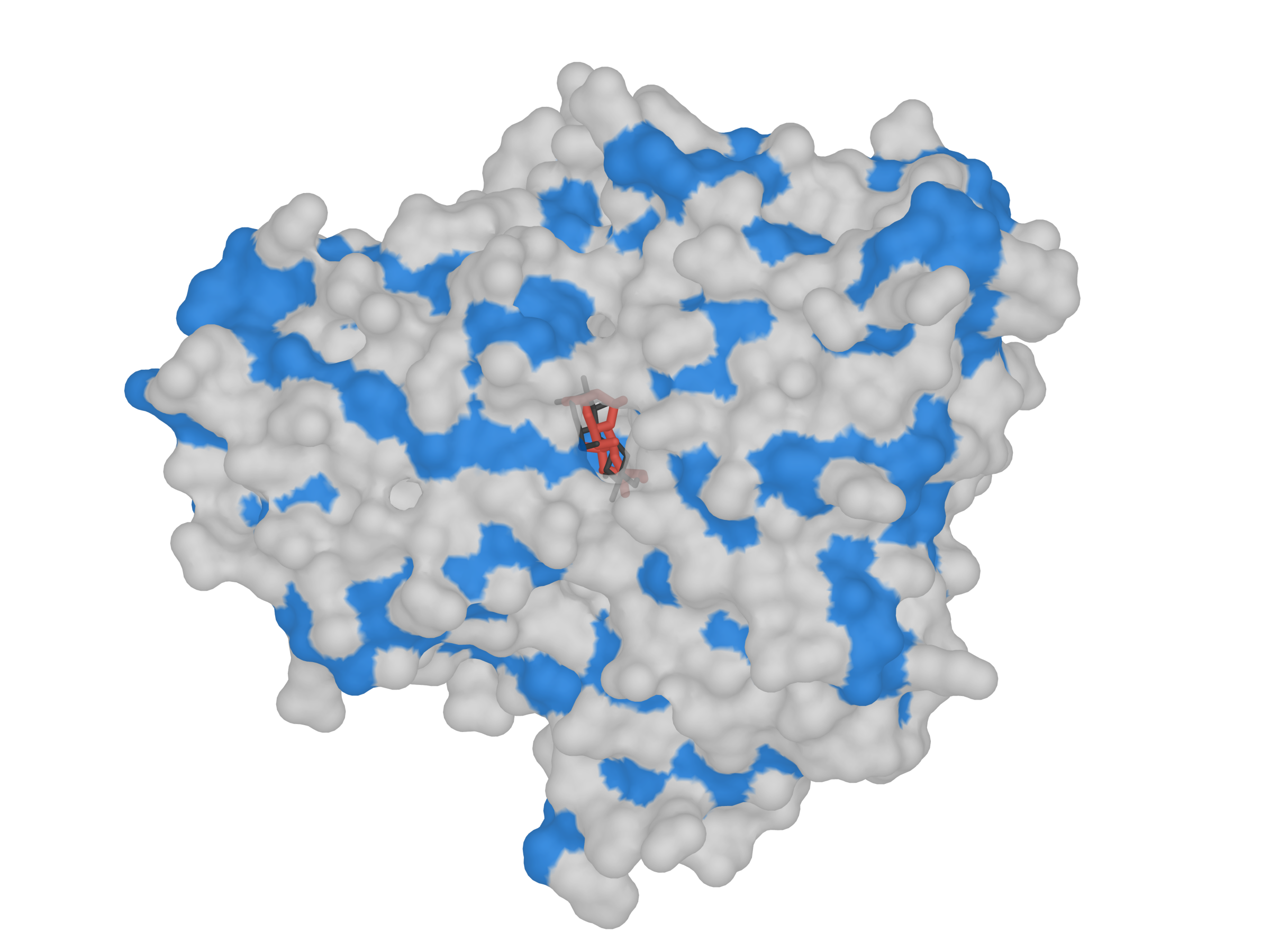} \\
    \small \textbf{6E6J}: L-RMSD \textbf{0.64 \AA}, AA-RMSD \textbf{0.77 \AA} &
    \small \textbf{6JAN}: L-RMSD \textbf{1.26 \AA}, AA-RMSD \textbf{0.90 \AA} \\[1.0em]

    \includegraphics[width=0.48\textwidth]{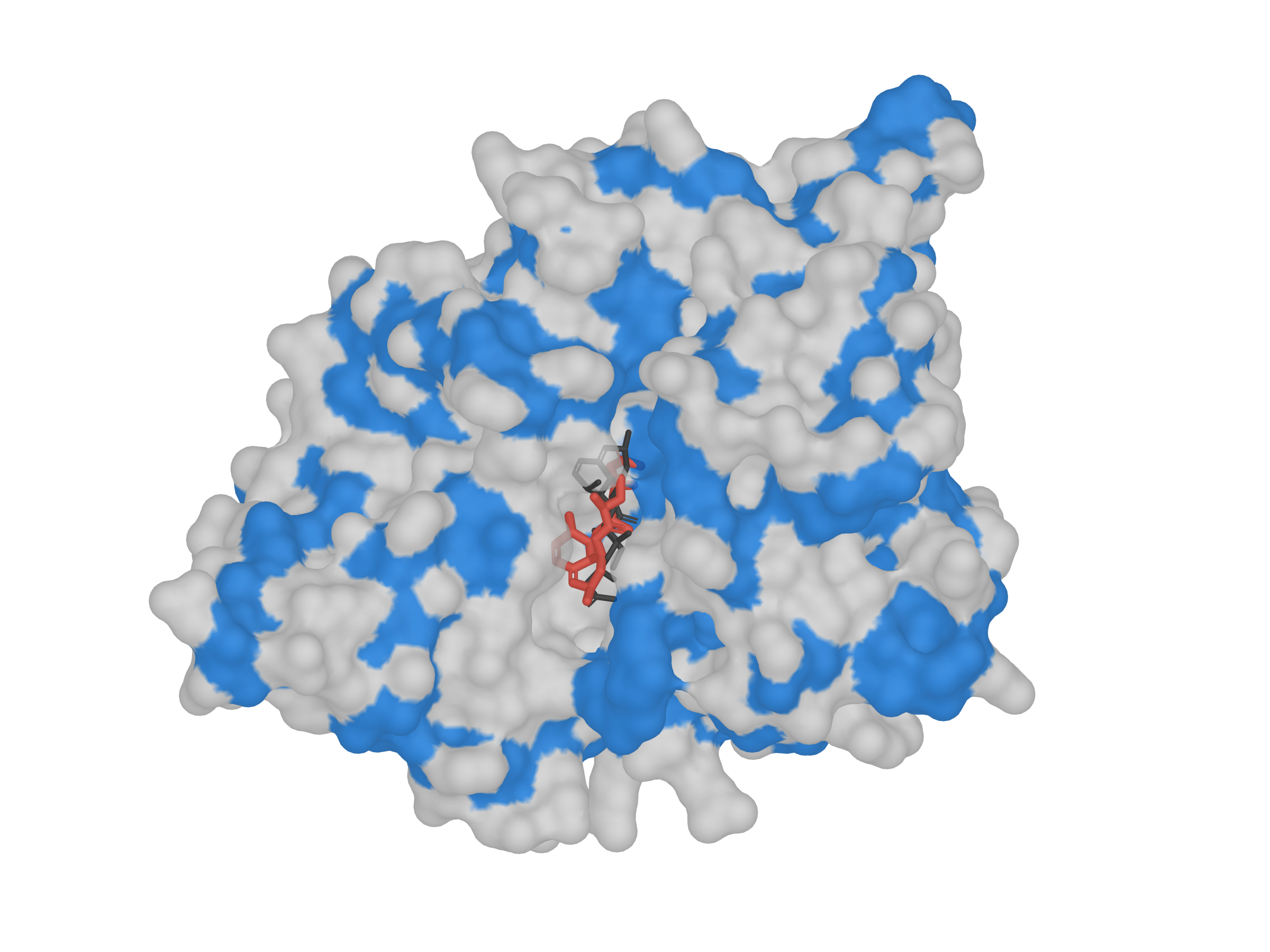} &
    \includegraphics[width=0.48\textwidth]{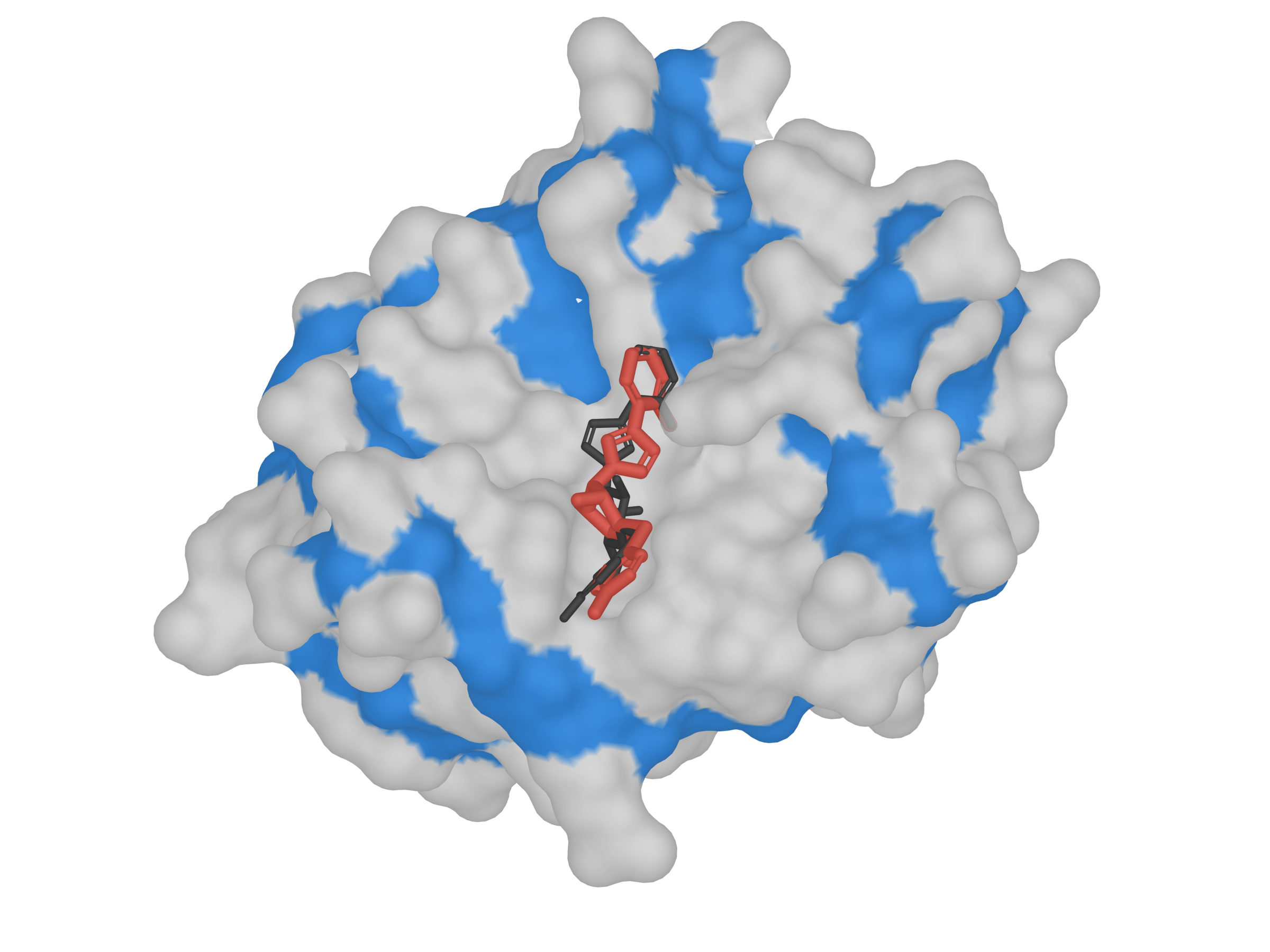} \\
    \small \textbf{6KJD}: L-RMSD \textbf{3.97 \AA}, AA-RMSD \textbf{3.04 \AA} &
    \small \textbf{6QLT}: L-RMSD \textbf{1.44 \AA}, AA-RMSD \textbf{1.00 \AA}
  \end{tabular}

  \caption{
    Extended results.
  }
  \label{fig:visualizations-grid}
\end{figure}

\newpage
\section*{NeurIPS Paper Checklist}

\begin{enumerate}

\item {\bf Claims}
    \item[] Question: Do the main claims made in the abstract and introduction accurately reflect the paper's contributions and scope?
    \item[] Answer: \answerYes{}{}%
    \item[] Justification: We provide comprehensive experiments regarding PDBBind, PoseBusters and claims from abstract in \Cref{subsec:exp-main}.

\item {\bf Limitations}
    \item[] Question: Does the paper discuss the limitations of the work performed by the authors?
    \item[] Answer: \answerYes{}%
    \item[] Justification: We include "Future Work and Limitations" discussion (see \Cref{sec:discussion}). 

\item {\bf Theory assumptions and proofs}
    \item[] Question: For each theoretical result, does the paper provide the full set of assumptions and a complete (and correct) proof?
    \item[] Answer: \new{\answerYes{}}%
    \item[] Justification: \new{Proofs of all propositions are provided in Appendix~\ref{app:proofs}.}

    \item {\bf Experimental result reproducibility}
    \item[] Question: Does the paper fully disclose all the information needed to reproduce the main experimental results of the paper to the extent that it affects the main claims and/or conclusions of the paper (regardless of whether the code and data are provided or not)?
    \item[] Answer: \answerYes{}%
    \item[] Justification: We provide a detailed overview of data preparation, model training and architecture, including pseudocodes in Appendix \ref{app:data-prep}, \ref{app:data}, \ref{app:model}, \ref{app:featurization}, \ref{app:exp-details}.

\item {\bf Open access to data and code}
    \item[] Question: Does the paper provide open access to the data and code, with sufficient instructions to faithfully reproduce the main experimental results, as described in supplemental material?
    \item[] Answer: \answerYes{}%
    \item[] Justification: Data we use is open-source and free to download. \new{The source code for training and evaluating our method is publicly available at \url{https://github.com/dunnolab/harmony}.}

\item {\bf Experimental setting/details}
    \item[] Question: Does the paper specify all the training and test details (e.g., data splits, hyperparameters, how they were chosen, type of optimizer) necessary to understand the results?
    \item[] Answer: \answerYes{}%
    \item[] Justification: We review data preparation and training details throughout Appendix sections.

\item {\bf Experiment statistical significance}
    \item[] Question: Does the paper report error bars suitably and correctly defined or other appropriate information about the statistical significance of the experiments?
    \item[] Answer: \new{\answerYes{}}%
    \item[] Justification: \new{We report mean and standard deviation over three training seeds for \textit{Harmony} and the retrained baselines (\Cref{tab:docking_results}).}

\item {\bf Experiments compute resources}
    \item[] Question: For each experiment, does the paper provide sufficient information on the computer resources (type of compute workers, memory, time of execution) needed to reproduce the experiments?
    \item[] Answer: \answerYes{}%
    \item[] Justification: We provide information regarding computational resources needed to train the model in \Cref{subsec:conf-aff-relax}.
    
\item {\bf Code of ethics}
    \item[] Question: Does the research conducted in the paper conform, in every respect, with the NeurIPS Code of Ethics \url{https://neurips.cc/public/EthicsGuidelines}?
    \item[] Answer: \answerYes{}%
    \item[] Justification: None

\item {\bf Broader impacts}
    \item[] Question: Does the paper discuss both potential positive societal impacts and negative societal impacts of the work performed?
    \item[] Answer: \answerNo{}%
    \item[] Justification: The paper introduces novel generative model for flexible docking. At the current stage of the field, method does not hold any negative societal impacts and can not be used to harm.
    
\item {\bf Safeguards}
    \item[] Question: Does the paper describe safeguards that have been put in place for responsible release of data or models that have a high risk for misuse (e.g., pre-trained language models, image generators, or scraped datasets)?
    \item[] Answer: \answerNA{}%
    \item[] Justification: None

\item {\bf Licenses for existing assets}
    \item[] Question: Are the creators or original owners of assets (e.g., code, data, models), used in the paper, properly credited and are the license and terms of use explicitly mentioned and properly respected?
    \item[] Answer: \answerYes{}%
    \item[] Justification: None

\item {\bf New assets}
    \item[] Question: Are new assets introduced in the paper well documented and is the documentation provided alongside the assets?
    \item[] Answer: \answerYes{}%
    \item[] Justification: None

\item {\bf Crowdsourcing and research with human subjects}
    \item[] Question: For crowdsourcing experiments and research with human subjects, does the paper include the full text of instructions given to participants and screenshots, if applicable, as well as details about compensation (if any)? 
    \item[] Answer: \answerNA{}%
    \item[] Justification: None

\item {\bf Institutional review board (IRB) approvals or equivalent for research with human subjects}
    \item[] Question: Does the paper describe potential risks incurred by study participants, whether such risks were disclosed to the subjects, and whether Institutional Review Board (IRB) approvals (or an equivalent approval/review based on the requirements of your country or institution) were obtained?
    \item[] Answer: \answerNA{}%
    \item[] Justification: None

\item {\bf Declaration of LLM usage}
    \item[] Question: Does the paper describe the usage of LLMs if it is an important, original, or non-standard component of the core methods in this research? Note that if the LLM is used only for writing, editing, or formatting purposes and does \emph{not} impact the core methodology, scientific rigor, or originality of the research, declaration is not required.
    \item[] Answer: \answerNA{}%
    \item[] Justification: None

\end{enumerate}

\end{document}